\documentclass[reqno]{amsart}

\usepackage[mathscr]{eucal} % \EuScript (\mathcal unchanged)
\usepackage{graphicx}       
\usepackage{xcolor}         
\usepackage{amsmath, amsfonts, amssymb, amsthm, amscd, mathabx}
\usepackage{epic, eepic}
\usepackage{longtable, array,subcaption}
\usepackage{here}
\usepackage{pdfpages}
\usepackage{ytableau}

\usepackage{tikz}
\usepackage{caption} 

\usepackage[colorlinks=true, pdfstartview=FitV, linkcolor=blue, citecolor=blue, urlcolor=blue]{hyperref}
\usepackage[capitalize,nameinlink,noabbrev,nosort]{cleveref}

\usepackage{geometry}
\numberwithin{equation}{section}

\newtheorem{theorem}{Theorem}[section]
\newtheorem{lemma}[theorem]{Lemma}

\newtheorem{proposition}[theorem]{Proposition}

\theoremstyle{definition}

\newtheorem{example}[theorem]{Example}

\newcommand{\Z}{{\mathbb Z}}

\newcommand{\ok}{{\rm{\bf k}}}

\newcommand{\am}{{\rm\aaa}^{\!-} }
\newcommand{\ap}{{\rm\aaa}^{\!+} }

\newcommand{\LL}{\mathcal{L}}
\newcommand{\BB}{\mathscr{B}}
\newcommand{\TT}{\mathscr{T}}
\newcommand{\vv}{{\mathsf v}}
\newcommand{\VV}{{\mathsf V}}
\newcommand{\Hp}{{H_{\text{PushTASEP} }}}

\newcommand{\Q}{\mathbb{Q}}

\newcommand{\aaa}{\mathbf{a}}

\DeclareMathOperator{\sgn}{sgn}

\begin{document}

\title[Multispecies $t$-PushTASEP]
{Multispecies inhomogeneous $t$-PushTASEP
\\  from antisymmetric fusion}

\author[Arvind Ayyer]{Arvind Ayyer}
\address{Arvind Ayyer, Department of Mathematics, Indian Institute of Science,
Bangalore 560012, India}
\email{arvind@iisc.ac.in}

\author[Atsuo Kuniba]{Atsuo Kuniba}
\address{Atsuo Kuniba, Graduate School of Arts and Sciences, University of Tokyo, Komaba, Tokyo, 153-8902, Japan}
\email{atsuo.s.kuniba@gmail.com}

\date{\today}

\dedicatory{Dedicated to the memory of Rodney James Baxter}

\begin{abstract}
We investigate the recently introduced 
inhomogeneous $n$-species $t$-PushTASEP, a long-range stochastic process on a periodic lattice.
A Baxter-type formula is established, expressing the Markov matrix as an alternating sum 
of commuting  transfer matrices over all the fundamental representations of $U_t(\widehat{sl}_{n+1})$.
This superposition acts as an inclusion-exclusion principle, selectively extracting the sequential particle 
transitions characteristic of the PushTASEP, while canceling forbidden channels.
The homogeneous specialization connects the PushTASEP to ASEP, 
showing that the two models share eigenstates and a common integrability structure. 
\end{abstract}

\subjclass[2020]{60J27, 82B20, 82B23, 82B44, 81R50, 17B37}
\keywords{$t$-PushTASEP, multispecies, inhomogeneous, Yang--Baxter equation}

\maketitle

\section{Introduction}

The totally asymmetric simple exclusion process (TASEP) is a stochastic model of interacting particles 
introduced around 1970s  in \cite{MGP68,S70}.
PushTASEP is a long-range variant where particles are allowed to hop to 
distant sites under certain rules.
A characteristic feature of its dynamics is the simultaneous movement of multiple pushed particles, 
triggered by the arrival of another particle.
Several variations of PushTASEP have been introduced and studied extensively
from the viewpoints of probability theory,  statistical mechanics, algebraic combinatorics, 
special functions, integrable systems, representation theory, etc.
See for example~\cite{ANP23,AM23,AMW24,BW22,CP13,M20,P19} and the references therein.

In this paper we focus on the version studied in \cite{AMW24}.
For a given positive integer $n$, each local state is selected from $\{0,1,\ldots, n\}$, where $1,\ldots, n$ represent 
the presence of one of the $n$ species of particles, and $0$ corresponds to an empty site. 
The system evolves under a long-range stochastic dynamics 
on a
one-dimensional periodic lattice of length $L$, with 
hopping rates that depend on a parameter $t$ and  also on 
$x_1,\ldots, x_L$, assigned to the lattice sites representing the inhomogeneity of the system.
We refer to it as the \emph{inhomogeneous $n$-species $t$-PushTASEP}, or simply \emph{PushTASEP}.

Let $H_{\text{PushTASEP}}(x_1,\ldots, x_L)$ denote its Markov matrix (see \eqref{Hdef}),
which appears in the continuous-time master equation.
It preserves a subspace specified by the number $m_i$ of particles of each type $i$. 
Set $K_i = m_0+\cdots + m_{i-1}$.
The main result of this paper, Theorem \ref{th:main}, is as follows:
\begin{multline}
\label{main}
H_{\text{PushTASEP}}(x_1,\ldots, x_L) = 
\frac{1}{(1-t)\prod_{i=1}^n(1-t^{K_i})}
\sum_{k=0}^{n+1}(-1)^{k-1}\left. \frac{dT^k(z|x_1,\ldots, x_L)}{dz}\right|_{z=0} \\
 -\left( \sum_{j=1}^L\frac{1}{x_j} \right)\mathrm{Id},
\end{multline}
where $T^0(z|x_1,\ldots, x_L), \ldots, T^{n+1}(z|x_1,\ldots, x_L)$  
are commuting transfer matrices of integrable two-dimension\-al vertex models in the sense of Baxter \cite{Bax82},
with spectral parameter $z$ and inhomogeneities $x_1,\ldots, x_L$:
\begin{align}\label{tcom0}
[T^k(z|x_1,\ldots, x_L), T^{k'}(z'|x_1,\ldots, x_L)]=0\qquad (0 \le k,k' \le n+1).
\end{align}
A novelty here lies in the fact that $T^k(z|x_1,\ldots, x_L)$ has the auxiliary space given by the degree $k$ 
{\em antisymmetric tensor representation} of the quantum affine algebra $U_t(\widehat{sl}_{n+1})$ in a certain gauge.
The corresponding quantum $R$ matrix is derived by the 
{\em antisymmetric fusion},  in contrast to the symmetric fusion adopted in 
almost all similar results obtained so far in the realm of integrable probability.\footnote{It is also derived, even
more simply, from the three-dimensional $L$-operator
satisfying the tetrahedron equation, as reviewed in Appendix \ref{app:3dR}.}

To further expand the perspective of the result \eqref{main}, 
let us also consider 
short range models, where the most extensively studied prototype is 
the asymmetric simple exclusion process (ASEP).
Specifically, we focus on the $n$-species ASEP on the one-dimensional periodic lattice of length $L$,
defined on the same state space as the aforementioned PushTASEP.
The ASEP exhibits an asymmetry in the adjacent hopping rates, specified by  the parameter $t$, 
but otherwise the system is homogeneous and possesses the $\mathbb{Z}_L$-translational symmetry.
A variety of results have been obtained regarding the stationary states of ASEP;  
see, for instance, \cite{ANP23,BW22, CDW15,CMW22,KOS24,M20, PEM09} and the references therein.
Let $H_\text{ASEP}$ denote the Markov matrix governing the continuous-time master equation
(see \eqref{hasep1}-\eqref{hasep2}).
It is well-known that the integrability of ASEP is attributed to the underlying commuting transfer matrices as
 \begin{align}\label{ht0}
 H_{\text{ASEP}} = -(1-t)\frac{d}{dz} \left. \log T^1(z|{\bf x}={\bf 1}) \right|_{z=1}-tL\, \mathrm{Id},
\end{align}
where $T^1(z|{\bf x}={\bf 1})$ is a summand corresponding to $k=1$ in \eqref{main},  and 
${\bf x}={\bf 1}$ indicates the specialization to the homogeneous case $x_1=\cdots = x_L=1$.
This kind of origin of the ``Hamiltonians" in the commuting transfer matrices 
is commonly referred to as Baxter's formula (cf.~\cite[eq.~(10.14.20)]{Bax82}).
As is customary, the evaluation is performed at the so-called ``Hamiltonian point", $z=1$ in the present setting,
where $T^1(z|{\bf x}={\bf 1})$ reduces to a simple lattice shift operator, and the Hamiltonian 
becomes a sum of adjacent interaction terms under the homogeneous  setting ${\bf x}={\bf 1}$.

Our formula \eqref{main} is a Baxter-type formula for long-range stochastic process models,
where such a  Hamiltonian point does {\em not} exist due to the inherent inhomogeneity of the system.
The most noteworthy feature of \eqref{main} is that it includes the {\em superposition} over 
all the transfer matrices corresponding to 
the fundamental representations of $U_t(\widehat{sl}_{n+1})$ for their auxiliary spaces.
This is particularly intriguing because 
the {\em individual} transfer matrix $T^k(z|x_1,\ldots, x_L)$ is generally {\em not} stochastic; 
it neither satisfies non-negativity 
nor the so-called sum-to-unity property (cf.~\cite[Sec. 3.2]{KMMO16}) in general.\footnote{There are few exceptions 
that can be made stochastic, including the cases $k=0,1,n+1$.}
The alternating sum in \eqref{main} operates as an inclusion-exclusion principle, selectively extracting the
allowed particle dynamics in the PushTASEP with proper transition rates, while  
dismissing all other unwanted channels. 
It would be interesting to investigate whether a similar mechanism is also effective in the generalized models 
where each site can accommodate more than one particle.
The summation over the fundamental representations corresponds to the dimension
$\binom{n+1}{0}+\binom{n+1}{1}+\cdots + \binom{n+1}{n+1}=2^{n+1}$.
It indicates a further reformulation, possibly through three-dimensional integrability 
(cf.  \cite[Chap.~18]{K22}), which is left, however, as a problem for future investigation.

Let $\Hp({\bf x}={\bf 1})$ denote  
$\Hp(x_1,\ldots, x_L)$ under the homogeneous choice  $x_1=\cdots = x_L=1$.
This specialization presents no subtlety.
The result \eqref{main} and \eqref{ht0} reveal that 
the homogeneous PushTASEP and ASEP are ``sister models", whose integrability originates from the 
same family of commuting transfer matrices 
$\{T^k(z|{\bf x}={\bf 1})\}$ corresponding to the fundamental representations.
A direct consequence of the Yang--Baxter commutativity \eqref{tcom0} is:
\begin{align}\label{hh0}
[H_\text{ASEP}, \Hp({\bf x}={\bf 1})]=0.
\end{align}
It follows that the two models share the same eigenstates. 
It was observed in \cite[Corollary 1.3]{AMW24} that these two models share the same stationary distribution, 
which was a key motivation of the work.
The property \eqref{hh0} 
provides a simple explanation for this coincidence 
and shows that the same stationary state is a {\em joint eigenstate} of all $T^k(z|x_1, \ldots, x_L)$.
It also gives rise to an interesting question; 
which one among the ASEP and the homogeneous PushTASEP mixes faster, i.e. converges faster to
their common stationary distribution starting from the same initial
condition.

Let us comment on the inhomogeneous $n$-species $t$-PushTASEP models {with arbitrary capacity, which was} investigated in \cite{BW22,ANP23}. They are formulated as integrable vertex
models associated with $U_t(\widehat{sl}_{n+1})$. 
The model in \cite[Section 6]{ANP23} is a higher capacity version of our model and is a special case of a stochastic vertex model introduced in \cite[Section 12.5]{BW22}. Although the transition rates in 
\cite[Section 6]{ANP23} look superficially different, they turn out to be equivalent to our model at capacity $1$ because of the periodic boundary conditions.
However, their results are of a very different flavor. 
Unlike this work, they do not give a closed form 
expression 
for the Markov matrix.

The layout of the paper is as follows.
In Section \ref{sec:tpush}, we provide a precise definition of the PushTASEP following \cite{AMW24}.
In Section \ref{sec:S}, we explain the antisymmetric fusion construction (cf. \cite{KRS81}) of the 
quantum $R$ matrix $S^{k,1}(z)$ for $0 \le k \le n+1$.
In Section \ref{sec:T},  we introduce 
the transfer matrices $T^k(z|x_1,\ldots, x_L)$ and describe their basic properties.
Section \ref{sec:H} constitutes the core of the paper.
It presents the main Theorem \ref{th:main} and its proof.
In Section \ref{sec:FP}, we provide further remarks on 
the eigenvalues of the transfer matrix and the matrix product formula for stationary states.
In Section \ref{sec:asep}, we include an analogous, but much simpler known result on ASEP
for reader's convenience.
Appendix \ref{app:proof} includes a proof of Theorem \ref{th:main}.
Appendix \ref{app:3dR} reviews a matrix product construction of the quantum
$R$ matrix $S^{k,1}(z)$ using the three-dimensional $L$-operator, with careful
adjustment of gauge and normalization. This approach, based on \cite{BS06} and
\cite[Chap.~11]{K22}, provides a simpler way to compute the matrix elements
than the standard fusion procedure.
Appendix \ref{app:sym} presents another formula for 
$H_\text{PushTASEP}(x_1,\ldots, x_L)$ in terms of transfer matrices associated with the symmetric fusion
for comparison.

\section{Multispecies $t$-PushTASEP}\label{sec:tpush}

\subsection{Definition of $n$ species inhomogeneous $t$-PushTASEP}

Let us recall the $n$ species inhomogeneous 
$t$-PushTASEP on one dimensional periodic lattice of length $L$ in \cite{AMW24}.
It is a continuous time Markov process on $\mathbb{V} =\VV^{\otimes L}$, where 
 $\VV = \bigoplus_{\sigma=0}^n \mathbb{C} \vv_\sigma$ denotes the space of local states.
 We will often write $\vv_{\sigma_1} \otimes \cdots \otimes \vv_{\sigma_L}$ simply as 
 $|\sigma_1,\ldots, \sigma_L\rangle$ or $|\boldsymbol{\sigma}\rangle$
 with an array 
 $\boldsymbol{\sigma} = (\sigma_1,\ldots, \sigma_L) \in \{0,\ldots, n\}^L$.
We regard a local state $v_\sigma$ as an empty site if $\sigma=0$ and the one 
occupied by a particle of type $\sigma$ for $1 \le \sigma \le n$. 

Let $\mathbb{V}({\bf m}) \subset \mathbb{V}$ be the subspace 
specified by the {\em multiplicity} ${\bf m} = (m_0,\ldots, m_n)\in (\Z_{\ge 0})^{n+1}$ 
of the particles as follows:
\begin{align}
\mathbb{V}({\bf m})
&=\bigoplus_{ (\sigma_1,\ldots, \sigma_L) \in \mathcal{S}({\bf m})}
\mathbb{C}  |\sigma_1,\ldots, \sigma_L\rangle,
\label{Vm}
\\
\mathcal{S}({\bf m}) &= \{(\sigma_1,\ldots, \sigma_L) \in \{0,\ldots, n\}^L\mid 
\delta_{i,\sigma_1}+\cdots + \delta_{i,\sigma_L}=m_i\, (0 \le i \le n)\}.
\label{Sm}
\end{align}
Note that $m_0 + \cdots + m_n = L$.
We set
\begin{align}
K_i& = m_0+ \cdots + m_{i-1} \quad (0 \le i \le n),
\label{Ki}\\
D_{\bf m} &= (1-t)\prod_{i=1}^n(1-t^{K_i}).\label{Dm}
\end{align}
By the definition $K_0 =0$.
We shall exclusively consider the case $m_0,\ldots, m_n\ge 1$ throughout the article, 
and hence $K_i \ge i$ for $1 \le i \le n$ and $D_{\bf m} \neq 0$.
 
The $n$ species inhomogeneous $t$-PushTASEP is a stochastic process on each $\mathbb{V}({\bf m})$ 
governed  by the master equation 
\begin{align}
\frac{d}{ds}|\mathbb{P}(s)\rangle = \Hp(x_1,\ldots, x_L)|\mathbb{P}(s)\rangle
\end{align}
for the state vector 
$|\mathbb{P}(s)\rangle  = \sum_{\boldsymbol{\sigma}}\mathbb{P}(\boldsymbol{\sigma};s)
|\boldsymbol{\sigma}\rangle$ with 
 the coefficient 
$\mathbb{P}(\boldsymbol{\sigma};s)$ being the probability of
the occurrence of the configuration 
$\boldsymbol{\sigma}$ at time $s$.
The Markov matrix 
$\Hp = \Hp(x_1,\ldots, x_L): \mathbb{V}({\bf m}) \rightarrow \mathbb{V}({\bf m})$ is defined by 
\begin{align}\label{Hdef}
\Hp |\boldsymbol{\sigma}\rangle &= 
\sum_{\substack{\boldsymbol{\sigma}' \in \mathcal{S}({\bf m}) \\ \boldsymbol{\sigma}' \neq \boldsymbol{\sigma}}} 
\sum_{j=1}^L \frac{1}{x_j}\prod_{1 \le h \le n}
w^{(j)}_{\boldsymbol{\sigma}, \boldsymbol{\sigma}'}(h)|\boldsymbol{\sigma}'\rangle
-\left(\sum_{j=1}^L\frac{[\sigma_j\ge 1]}{x_j} \right)|\boldsymbol{\sigma}\rangle,
\end{align}
where we employ the Iverson bracket $[\text{true}]=1, [\text{false}]=0$ throughout.
The parameter $x_j>0$ is associated with the lattice site $j \in \{1,\ldots, L\}$, and 
represents the inhomogeneity of the system at that site.
The factor $w^{(j)}_{\boldsymbol{\sigma}, \boldsymbol{\sigma}'}(h)$, which constitutes the core part of $\Hp$,  
is a rational function of $t$ described in \cite[sec. 2.2]{AMW24}. 
For readers' convenience, we recall its definition below.
{We first give the informal definition, followed by the formal one.}

{
The configurations of the $t$-PushTASEP are elements in $\mathcal{S}(m)$, to be thought of as cyclic strings.
The particle $b$ at site $j$ gets \emph{activated} at rate $1/x_j$, and then starts to move clockwise, i.e. forward. When $b$ finds a particle weaker than it (i.e. with smaller value) at a site $k$ say, $b$ settles there and `activates' the weaker particle with probability $1-t$ and continues to move forward with probability $t$. Since the lattice is finite with periodic boundary conditions, $b$ will find a particle to activate with probability $1$. In particular, if there are $z$ weaker particles, $b$ will displace the $i$'th one in order with probability $t^{i-1}(1-t)/(1-t^z)$.
The activated particle continues moving forward and activating particles in the same way. This ends when an activated particle settles at a vacancy. All of these moves are part of a single transition and occur instantaneously.
See Figure \ref{fig:eg trans} for an example.
}

\begin{figure}[h!]
\begin{center}
\includegraphics[scale=1]{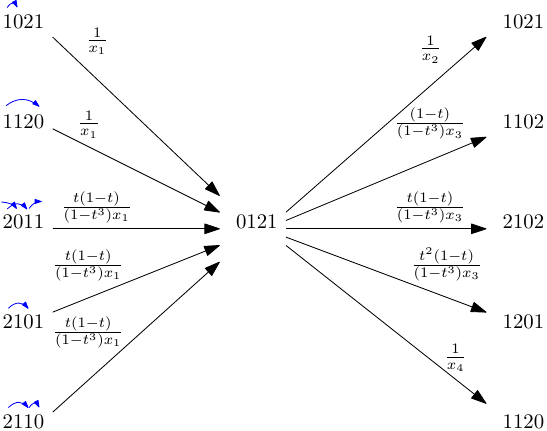}
\caption{{All transitions going in and out of a configuration in $\mathcal{S}(1,2,1)$. The moves for the incoming transitions are shown with blue arrows.}
\label{fig:eg trans}}
\end{center}
\end{figure}

Let $\boldsymbol{\sigma}, \boldsymbol{\sigma}' \in \mathcal{S}({\bf m})$.
Then $w^{(j)}_{\boldsymbol{\sigma}, \boldsymbol{\sigma}'}(h)$ is defined to be zero 
except when the following conditions are satisfied: 

\begin{itemize}
\item $j$ is the unique site such that $\sigma_j \in \{1,\ldots, n\}$ and $\sigma'_j =0$. For every other site $i$,
$\sigma_i \le \sigma'_i$.

\item For each type $h \in \{1,\ldots, n\}$, either:
\begin{enumerate}
\item the sites occupied by species $h$ are the same in $\boldsymbol{\sigma}$ and $\boldsymbol{\sigma}'$; or,

\item there exists exactly one site $p(h)$ such that $\sigma_{p(h)}=h$ and $\sigma'_{p(h)}\neq h$.
It follows that there also exists exactly one site $p'(h)$ such that $\sigma'_{p'(h)}=h$ and $\sigma_{p'(h)} \neq h$.
\end{enumerate}
\end{itemize}
If case 1 holds, then $w^{(j)}_{\boldsymbol{\sigma}, \boldsymbol{\sigma}'}(h)=1$.
If case 2 holds, then let $\ell_h$ be the number of sites in the cyclic interval $(p(h), p'(h))$, 
excluding endpoints, with value smaller than $h$ in $\boldsymbol{\sigma}$.\footnote{The corresponding 
phrase  ``... smaller than $h$ in $\boldsymbol{\sigma}'$"  in \cite{AMW24} is a misprint.}
Then $w^{(j)}_{\boldsymbol{\sigma}, \boldsymbol{\sigma}'}(h)$ is defined as
\begin{align}\label{wdef}
w^{(j)}_{\boldsymbol{\sigma}, \boldsymbol{\sigma}'}(h)
= \frac{(1-t)t^{\ell_h}}{1-t^{K_h}}
\end{align}
using $K_h$ in \eqref{Ki}.
By these definitions, the first term in \eqref{Hdef} only contains the non-diagonal terms 
$|\boldsymbol{\sigma}'\rangle$ with $\boldsymbol{\sigma}' \neq \boldsymbol{\sigma}$.

\begin{example}\label{ex:H}
We consider the {example in Figure~\ref{fig:eg trans}, which is the case of
$m = (1,2,1)$,} $n=2$ and $L=4$.
Then 
\begin{equation}
\begin{split}
\Hp|0121\rangle =& \frac{|1021\rangle}{x_2}
+\frac{(1-t)|1102\rangle}{(1-t^3)x_3}
+\frac{t(1-t)|2101\rangle}{(1-t^3)x_3}
+\frac{t^2(1-t)|1201\rangle}{(1-t^3)x_3}
+\frac{|1120\rangle}{x_4}
\\
& -\left( \frac{1}{x_2}+\frac{1}{x_3}+\frac{1}{x_4} \right) |0121\rangle.
\end{split}
\end{equation}
\end{example}

\section{The matrix $S(z)$}\label{sec:S}

\subsection{Space $V^k$ with base labeled with $\BB^k$ and $\TT^k$ }

For $0 \le k \le n+1$, set 
\begin{align}
\BB^k &= \{{\bf i} = (i_0,\ldots, i_n) \in \{0,1\}^{n+1}\mid |{\bf i}|=k\},
\quad  (|{\bf i}|=i_0+\cdots + i_n),
\label{Bk}\\
V^k &= \bigoplus_{{\bf i} \in \BB^k} \mathbb{C}v_{\bf i}.
\label{Vk}
\end{align}
One has $\dim V^k = \binom{n+1}{k}$.
For the special case $k=1$, we identify $V^1$ with $\VV$, the space of local states of the $t$-PushTASEP,  
via\footnote{In Section \ref{ss:fusion}, 
$v_{\bf i} \in V^k$ with general  $k$ will be identified with the antisymmetric tensor rather than the simple monomial 
$\vv_{i_1}\otimes \cdots \otimes \vv_{i_k} \in \VV^{\otimes k}$.  
However such a connection is used only to explain the fusion procedure there.} 
\begin{align}\label{v11}
\VV \ni \vv_i = v_{{\bf e}_i} \in V^1 \;\;\text{where}\; \;{\bf e}_i=(\delta_{0,i},\ldots, \delta_{n,i}) \in \BB^1
\quad (0 \le i \le n).
\end{align}
Let us further introduce 
\begin{align}\label{tk}
\mathscr{T}^k = \{{\bf I} = (I_1,\ldots, I_k) \mid 0 \le I_1 < \cdots < I_k \le n\},
\end{align}
which we regard as the set of depth $k$ column strict (standard) tableaux 
over the alphabet $0,\ldots, n$. 
For example, with $n = 3$,
\[
\mathscr{T}^2 = \left\{
\begin{ytableau}
0 \\
1
\end{ytableau} \;, \quad
\begin{ytableau}
0 \\
2
\end{ytableau} \;, \quad
\begin{ytableau}
0 \\
3
\end{ytableau} \;, \quad
\begin{ytableau}
1 \\
2
\end{ytableau} \;, \quad
\begin{ytableau}
1 \\
3
\end{ytableau} \;, \quad
\begin{ytableau}
2 \\
3
\end{ytableau}
\right\}.
\]
We identify $\BB^k$ and $\mathscr{T}^k$ by the one-to-one correspondence 
where $i_\alpha=0,1 $ in  ${\bf i} \in \BB^k$ is regarded as the multiplicity of 
the letter $\alpha$ in ${\bf I} \in \mathscr{T}^k$. 
The arrays ${\bf i}$ in \eqref{Bk} and ${\bf I} \in \mathscr{T}^k$ in \eqref{tk}
will be referred to as the {\em multiplicity representation}  
and the {\em tableau representation}, respectively.

\subsection{$S^{k,1}(z)$ from antisymmetric fusion}\label{ss:fusion}

Here we present the matrix $S^{k,1}(z) \in \mathrm{End}(V^k\otimes V^1)$
with $0 \le k \le n+1$ satisfying the Yang-Baxter relation.
For simplicity, we shall often suppress the superscripts and denote it by $S(z)$.
Its action on the base vectors of $V^k \otimes V^1$ is given by 
\begin{align}
S(z)(v_{\bf i} \otimes v_{{\bf e}_j}) &=\sum_{{\bf a}\in \BB^k,  {\bf e}_b \in \BB^1}
S(z)^{{\bf a}, {\bf e}_b}_{{\bf i}, {\bf e}_j}\,v_{\bf a} \otimes v_{{\bf e}_b}
\qquad ({\bf i} \in \BB^k, {\bf e}_j \in \BB^1).
\label{S1n}
\end{align}
The matrix element will be depicted as in Figure~\ref{fig:matrix elem}.

\begin{figure}[h]
\centering
\begin{tikzpicture}

\node at (-1.5,-0.03){$\displaystyle{S(z)^{{\bf a}, {\bf e}_b}_{{\bf i}, {\bf e}_j}} = $};
  % vertical 
  \draw[->] (0.6,-0.6) node[below]{$j$}-- (0.6,0.6) node[above]{$b$};
  % horizontal
  \draw[->,line width=0.6mm] (0,0) node[left]{${\bf i}$}-- (1.2,0) node[right]{${\bf a}$};
  \draw(0.6,-0.25) arc[start angle=-90, end angle=-180, radius=0.25cm];
  \node at (0.2,-0.45) {$z$};
 
\end{tikzpicture}
\caption{A typical matrix element.}
\label{fig:matrix elem}
\end{figure}

\noindent
The horizontal arrow for $V^k$ and the vertical one for $V^1=\VV$ are distinguished by 
thick and ordinary arrows, respectively.
Note that in the diagram we use $j, b \in \{0,\ldots, n\}$ whereas ${\bf i}, {\bf a} \in \BB^k$.

The $R$-matrix $S^{k,1}(z)$ can be constructed in two ways.
A traditional method is the antisymmetric fusion of the $k=1$ case $S^{1,1}(z)$.
We will explain it in this subsection.
Another method, which is less known but actually much simpler and efficient, is based on 3D integrability.
It is summarized in  Appendix \ref{app:3dR}.

We write the elements $S(z)^{{\bf e}_a, {\bf e}_b}_{{\bf e}_i, {\bf e}_i} $
of $S^{1,1}(z)$ simply as $S(z)^{a,b}_{i,j}$.
They are zero except the following:
\begin{align}\label{Sb2}
S(z)^{i, i}_{i,i} = 1-tz,
\quad
S(z)^{i, j}_{i,j} = (1-z)t^{[i>j]} \; (i\neq j),
\quad
S(z)^{j,i}_{i,j} = (1-t)z^{[i<j]}\; (i\neq j).
\end{align}
They will be depicted as Figure \ref{fig:sb}.

\begin{figure}[h]
\begin{center}
\begin{picture}(320,50)(-20,30)
{\unitlength 0.011in
\put(12,80){
\put(-11,0){\vector(1,0){23}}\put(0,-10){\vector(0,1){22}}
\put(-5,0){\line(1,-1){5}\put(1,-11){$z$}}
}
\multiput(100,80)(90,0){3}{
\put(-11,0){\vector(1,0){23}}\put(0,-10){\vector(0,1){22}} 
\put(-5,0){\line(1,-1){5}
\put(2,-11){$z$}}
}
\put(-68,0){
\put(60.5,77){$i$}\put(77.5,60){$j$}\put(96,77){$a$}\put(77.5,96.5){$b$}
}
\put(20,0){
\put(60.5,77){$i$}\put(77.5,60){$i$}\put(96,77){$i$}\put(77.5,96.5){$i$}
}
\put(110,0){
\put(60.5,77){$i$}\put(77.5,60){$j$}\put(96,77){$i$}\put(77.5,96.5){$j$}
}
\put(200,0){
\put(60.5,77){$i$}\put(77.5,60){$j$}\put(96,77){$j$}\put(77.5,96.5){$i$}
}
\put(78,40){
\put(-84,0){$S(z)^{a,b}_{i,j}$}
\put(5,0){$1-tz$} 
\put(85,0){$(1-z)t^{[i>j]}$} \put(180,0){$ (1-t)z^{[i<j]}$}
}}
\end{picture}
\caption{Nonzero elements of $S^{1,1}(z)$ with $0 \le  i  \neq j \le n$.
Horizontal arrows, corresponding to $k=1$, are depicted as ordinary thin ones.}
\label{fig:sb}
\end{center}
\end{figure}

Recall that $V^1$ has been identified with the space of local states of the 
$t$-PushTASEP $\VV = \bigoplus_{i=0}^n \mathbb{C} \vv_i$ as  in \eqref{v11}.
Then the above definition is rephrased as  
\[
S(z)  = \sum_{a,b,i,j=0}^n S(z)^{a,b}_{i,j}E_{ai} \otimes E_{bj}
\in \mathrm{End}(\VV \otimes \VV),
\] 
where $E_{ij}\vv_l = \delta_{j,l}\vv_i$.
It satisfies the Yang-Baxter equation
\[
S_{1,2}(x)S_{1,3}(y)S_{2,3}(y/x) = 
S_{2,3}(y/x)S_{1,3}(y)S_{1,2}(x).
\]
Let $P(\mathsf{u} \otimes \mathsf{v}) = \mathsf{v} \otimes \mathsf{u}$  be the transposition.
From \eqref{Sb2}, one sees that the image 
\[
\mathrm{Im} PS(t^{-1}) 
=\bigoplus_{0 \le i < j \le n}\mathbb{C}(\vv_i \otimes \vv_j  - \vv_j \otimes \vv_i)
\] 
is the space of antisymmetric tensors.
 The Yang-Baxter equation multiplied with $P_{1,2}$ from the left reads
$P_{1,2}S_{1,2}(x)S_{1,3}(y)S_{2,3}(y/x) = 
S_{1,3}(y/x)S_{2,3}(y)P_{1,2}S_{1,2}(x)$.
The choice $x=t^{-1}$ here implies that 
\[
S_{1,3}(ty)S_{2,3}(y) \in \mathrm{End}(\VV \otimes \VV \otimes \VV)
\] 
preserves the space 
\[
\left(\mathrm{Im} P_{1,2}S_{1,2}(t^{-1}) \right) \otimes \VV 
=\left(\bigoplus_{0 \le i < j \le n}\mathbb{C}(\vv_i \otimes \vv_j  - \vv_j \otimes \vv_i)\right) \otimes \VV
\subset \VV \otimes \VV \otimes \VV.
\]

Consider the following operator with an overall scalar factor $d_k(z)^{-1}$ which is 
included to validate the forthcoming Theorem \ref{th:S}: 
\begin{align}
&d_k(z)^{-1}
S_{1,k+1}(z)S_{2,k+1}(t^{-1}z)\cdots S_{k,k+1}(t^{-k+1}z) \in \mathrm{End}(\VV^{\otimes k} \otimes \VV)
\quad (0 \le k \le n+1),
\label{Sk}
\\
&d_k(z) = (1-t^{-1}z)(1-t^{-2}z)\cdots (1-t^{-k+1}z).
\label{dk}
\end{align}
By extending the above argument, one can show that \eqref{Sk} 
can be restricted to $\bigwedge^k \otimes \VV$, where 
$\bigwedge^k$ is the subspace of $\VV^{\otimes k}$ spanned by the degree $k$ antisymmetric tensors
\begin{align}
\bigwedge\nolimits^{\! k}&= \bigoplus_{{\bf I} \in \mathscr{T}^k}\mathbb{C} \,\nu_{\bf I},
\qquad 
\nu_{\bf I}  = \sum_{\sigma \in \mathfrak{S}_k} \mathrm{sgn}(\sigma)\vv_{I_{\sigma(k)}}
\otimes \cdots \otimes \vv_{I_{\sigma(1)}} \in \VV^{\otimes k},
\label{La}
\end{align}
where $\mathfrak{S}_k$ is the symmetric group of degree $k$, and 
$\mathrm{sgn}(\sigma)$ denotes the signature of the permutation $\sigma$.
The set $\TT^k$ has been defined in \eqref{tk}.
We identify $\bigwedge^k$ in \eqref{La}  with
$V^k$ in \eqref{Vk} via 
\begin{align}\label{VLa}
&V^k \ni v_{\bf i} = \nu_{\bf I} \in \bigwedge\nolimits^{\! k},
\quad \BB^k \ni {\bf i} = {\bf I} \in \mathscr{T}^k,
\end{align}
where the bijective correspondence between the multiplicity arrays in $\BB^k$ \eqref{Bk} and 
the column strict tableaux in $\TT^k$ has been explained after \eqref{tk}.

The $R$-matrix $S^{k,1}(z)$ is obtained as the restriction of the operator \eqref{Sk}  
to $\bigwedge^k\otimes \VV$ according to the above identification.
This construction leads to the following formula for its matrix elements:
\begin{multline}
\label{Sf}
S(z)^{{\bf a}, {\bf e}_b}_{{\bf i}, {\bf e}_j}
=d_k(z)^{-1}\sum_{b_1,\ldots, b_{k-1} \in \{0,\ldots, n\}}
\sum_{\sigma \in \mathfrak{S}_k} \mathrm{sgn}(\sigma)
\\
\times S(z)^{A_k,b}_{I_{\sigma(k)},b_{k-1}}
S(t^{-1}z)^{A_{k-1},b_{k-1}}_{I_{\sigma(k-1)},b_{k-2}}
\cdots 
S(t^{-k+1}z)^{A_1,b_1}_{I_{\sigma(1)},j},
\end{multline}
where $b,j \in \{0,\ldots, n\}$ and 
${\bf a}, {\bf i}\in \BB^k$ are multiplicity arrays.
The arrays $(A_1,\ldots, A_k)$ and $(I_1,\ldots, I_k)  \in \mathscr{T}^k$ \eqref{tk} are the 
tableau representations of ${\bf a}$ and ${\bf i} \in \BB^k$, respectively.
We note that the sign factor  $\mathrm{sgn}(\sigma)$ is a simplifying feature of the current gauge, 
in contrast to the $t$-dependent factor  $(-t)^{\mathrm{length}(\sigma)}$, 
which is commonly encountered in the conventional gauge. 
To summarize, the weight is given as in Figure \ref{fig:matrix_element}.

\begin{figure}[h]
\centering
\begin{tikzpicture}

\node at (-1.6,-0.03){$\displaystyle{S(z)^{{\bf a}, {\bf e}_b}_{{\bf i}, {\bf e}_j}} \, = $};
  % vertical 
  \draw[->] (0.3,-0.6) node[below]{$j$}-- (0.3,0.6) node[above]{$b$};
  % horizontal
  \draw[->,line width=0.6mm] (-0.4,0) node[left]{${\bf i}$}-- (0.8,0) node[right]{${\bf a}$};
  \draw(0.3,-0.25) arc[start angle=-90, end angle=-180, radius=0.25cm];
  \node at (0.0,-0.45) {$z$};
  
  \node at (4.0, -0.1) {$\displaystyle{ = \; \frac{1}{\prod_{r=1}^{k-1} (1 - t^{-r} z)} 
  			\sum_{\sigma \in \mathfrak{S}_k} \sgn(\sigma)} \; \times \;\;\;$};
  \draw[->] (8.1,-1.5) node[below]{$j$}-- (8.1,1.6) node[above]{$b$};
  \draw[->] (7.5,-0.9) node[left]{$\sigma(I_1)$}-- (8.7,-0.9) node[right]{$A_1$};
  \draw(8.1,-1.15) arc[start angle=-90, end angle=-180, radius=0.25cm];
  \node at (7.4,-1.25) {${\scriptstyle z t^{-k+1}}$};
  \draw[->] (7.5,0.1) node[left]{$\sigma(I_2)$}-- (8.7,0.1) node[right]{$A_2$};
  \draw(8.1,-.15) arc[start angle=-90, end angle=-180, radius=0.25cm];
  \node at (7.4,-.30) {${\scriptstyle z t^{-k+2}}$};
  \node at (7.0,0.8) {$\vdots$};
  \node at (8.9,0.8) {$\vdots$};
  \draw[->] (7.5,1.3) node[left]{$\sigma(I_k)$}-- (8.7,1.3) node[right]{$A_k$};
  \draw(8.1,1.05) arc[start angle=-90, end angle=-180, radius=0.25cm];
  \node at (7.7,1.05) {${\scriptstyle z}$};
  
\end{tikzpicture}
\caption{The weight of the element in \eqref{Sf}, where $b_r$ is assigned to the 
 vertical edge between $A_r$ and $A_{r+1}$. 
Note that $I_1 < I_2 < \dots < I_k$ and $A_1 < A_2 < \dots < A_k$ by \eqref{tk}.}
\label{fig:matrix_element}
\end{figure}

We use the notation $[i, j] = \{i, i+1, \dots, j-1, j\}$ for $i\le j$.

\begin{theorem}\label{th:S}
\begin{align}
S(z)^{{\bf a}, {\bf e}_b}_{{\bf i}, {\bf e}_j} &= \delta^{{\bf a}+{\bf e}_b}_{{\bf i} + {\bf e}_j}
(-1)^{a_0+\cdots + a_{j-1}+i_0+\cdots + i_{b-1}}
 t^{a_{j+1}+\cdots + a_n}(1-t^{a_j}z^{\delta_{b,j}})z^{[j>b]}.
 \label{Srepeat0}
 \end{align}
\end{theorem}

We include a proof in Appendix \ref{app:proof}.
The sign factor in \eqref{Srepeat0} can also be expressed as
$(-1)^{a_{r+1}+\cdots + a_{s-1}}$ with 
$r=\min(j,b)$ and $s=\max(j,b)$.
The property that $S(z)^{{\bf a}, {\bf e}_b}_{{\bf i}, {\bf e}_j} =0$ unless 
${\bf a} + {\bf e}_b = {\bf i} + {\bf e}_j$ is referred to as weight conservation.
Matrix elements \eqref{Srepeat0} have appeared  in a different gauge 
as the basic ingredient in the vertex operator approach in \cite{DO94}.
The basic case $k=1$ in \eqref{Sb2} are positive in some range of $t, z$, and 
the sum $\sum_{0\le a,b\le n}S(z)^{a,b}_{i,j}=1-tz$ 
is independent of $i,j$.
It is well known that these properties can be utilized to construct a Markov Matrix of multispecies ASEP; see Section \ref{sec:asep}.

On the other hand for $k\neq 1$ in general, 
\eqref{Srepeat0} is neither positive definite nor negative definite for fixed $t$ and $z$.
Furthermore, the summation  
$\sum_{{\bf a} \in \BB^k, {\bf e}_b \in \BB^1} S(z)^{{\bf a}, {\bf e}_b}_{{\bf i}, {\bf e}_j}$
does {\em not} become independent of  ${\bf i}, {\bf e}_j$.
Consequently, the $R$-matrix $S^{k,1}(z)$  is {\em not stochastic} by itself in the sense of \cite{KMMO16}.
What is intriguing, as revealed by our subsequent analysis, is that
the Markov Matrix of the multispecies $t$-PushTASEP is nonetheless reproduced as 
a suitable {\em linear combination} of the transfer matrices constructed from 
$S^{0,1}(z),S^{1,1}(z), \ldots, S^{n+1,1}(z)$.

The matrix $S^{k,1}(z)$ given above is member of a larger family of $R$-matrices 
$S^{k,l}(z) \in \mathrm{End}(V^k \otimes V^l)$ with 
$0 \le k,l \le n+1$, and they satisfy the Yang-Baxter equation
\begin{align}\label{ybes}
S^{k_1,k_2}_{1,2}(x) S^{k_1,k_3}_{1,3}(xy) S^{k_2,k_3}_{2,3}(y) = 
S^{k_2,k_3}_{2,3}(y) S^{k_1,k_3}_{1,3}(xy) S^{k_1,k_2}_{1,2}(x).
\end{align}
An efficient method to construct $S^{k,1}(z)$ by 3D integrability 
is explained in Appendix \ref{app:3dR}. 

\section{Transfer matrix $T^k(z)$}\label{sec:T}

\subsection{Definition}\label{ss:Tdef}
Recall that $\mathbb{V}=\VV^{\otimes L}$ and we have identified $\VV$ with $V^1$ as in \eqref{v11}. 
Define the transfer matrix 
$T^k(z) = T^k(z|x_1,\ldots, x_L) : \mathbb{V} \longrightarrow  \mathbb{V}$ on the length $L$ periodic lattice by 
\begin{align}
T^k(z)= \mathrm{Tr}_{V^k}\left(S_{0,L} \Bigl( \frac{z}{x_L} \Bigr)
\cdots S_{0,1} \Bigl( \frac{z}{x_1} \Bigr)\right)
\qquad  (0 \le k \le n+1),
\end{align}
where  the index 0  denotes the auxiliary space $V^k$ over which the trace is taken.
The factor $S_{0,r}(z/x_r)$ is the matrix $S^{k,1}(z/x_r)$ defined by \eqref{S1n} and \eqref{Srepeat0}, 
which acts on $V^k \otimes (\text{$r$'th component of $\mathbb{V}$ from the left})$.
Explicitly, one has
\begin{subequations}
\begin{align}
T^k(z) |\sigma_1,\ldots, \sigma_L\rangle  &= \sum_{\sigma'_1,\ldots, \sigma'_L \in \{0,\ldots, n\}}
T^k(z)^{\sigma'_1,\ldots, \sigma'_L}_{\sigma_1,\ldots, \sigma_L}
|\sigma'_1, \ldots, \sigma'_L\rangle,
\label{tkdef}\\
T^k(z)^{\sigma'_1,\ldots, \sigma'_L}_{\sigma_1,\ldots, \sigma_L}
&= \sum_{{\bf a}_1,\ldots, {\bf a}_L \in \BB^k}
S \Bigl( \frac{z}{x_1} \Bigr)^{{\bf a}_2, {\bf e}_{\sigma'_1}}_{{\bf a}_1, {\bf e}_{\sigma_1}}
S \Bigl( \frac{z}{x_2} \Bigr)^{{\bf a}_3, {\bf e}_{\sigma'_2}}_{{\bf a}_2, {\bf e}_{\sigma_2}} 
\cdots 
S \Bigl( \frac{z}{x_L} \Bigr)^{{\bf a}_1, {\bf e}_{\sigma'_L}}_{{\bf a}_L, {\bf e}_{\sigma_L}}.
\label{tke}
\end{align}
\end{subequations}
We write the element \eqref{tke} as $\langle \boldsymbol{\sigma}'|T^k(z)| \boldsymbol{\sigma}\rangle$, 
and  depict it as Figure \ref{fig:tk}.
\begin{figure}[h]
\centering
\begin{tikzpicture}

\node at (-1.8,-0.18){$\displaystyle{\sum_{\phantom{AA}{\bf a}_1,\ldots, {\bf a}_L \in \BB^k}}$};
  % vertical 
  \draw[->] (0.6,-0.6) node[below]{$\sigma_1$}-- (0.6,0.6) node[above]{$\sigma'_1$};
  % horizontal
  \draw[->,line width=0.6mm] (0,0) node[left]{${\bf a}_1$}-- (1.2,0) node[right]{${\bf a}_2$};
  \draw(0.6,-0.25) arc[start angle=-90, end angle=-180, radius=0.25cm];
  \node at (0.2,-0.45) {$\frac{z}{x_1}$};
\begin{scope}[shift={(1.8,0)}]
  \draw[->] (0.6,-0.6) node[below]{$\sigma_2$}-- (0.6,0.6) node[above]{$\sigma'_2$};
  % horizontal
  \draw[->,line width=0.6mm] (0,0) -- (1.2,0) node[right]{${\bf a}_3$};
  \draw(0.6,-0.25) arc[start angle=-90, end angle=-180, radius=0.25cm];
  \node at (0.2,-0.45) {$\frac{z}{x_2}$};
   \end{scope}

  \draw[->,line width=0.6mm] (3.6,0) -- (4.4,0);
  \node at (4.8,0) {$\cdots$};
  \draw[->,line width=0.6mm] (5.1,0) -- (5.9,0);
   
   \begin{scope}[shift={(6.6,0)}]
     \draw[->] (0.6,-0.6) node[below]{$\sigma_L$}-- (0.6,0.6) node[above]{$\sigma'_L$};
  % horizontal
  \draw[->,line width=0.6mm] (0,0) node[left]{${\bf a}_L$}-- (1.2,0) node[right]{${\bf a}_1$};
  \draw(0.6,-0.25) arc[start angle=-90, end angle=-180, radius=0.25cm];
  \node at (0.2,-0.45) {$\frac{z}{x_L}$};
     \end{scope}
\end{tikzpicture}
\caption{Diagram representation of the matrix element 
$\langle \boldsymbol{\sigma}'|T^k(z)| \boldsymbol{\sigma}\rangle$.}
\label{fig:tk}
\end{figure}

The parameter $z$ is referred to as the \emph{spectral parameter},
while $x_1,\ldots, x_L$ represent
the inhomogeneity associated with the vertices.
Adopting the terminology from the box-ball systems \cite{IKT12},
we refer to the ${\bf a}_1, \ldots, {\bf a}_L \in \BB^k$ as {\em carriers} with capacity $k$.

\subsection{Basic properties}

From the Yang-Baxter relation \eqref{ybes} with $(k_1,k_2,k_3)=(k,k',1)$,  
one can show the commutativity
\begin{align}\label{tcom}
[T^k(z|x_1,\ldots, x_L), T^{k'}(z'|x_1,\ldots, x_L)]=0
\qquad (0 \le k, k' \le n+1).
\end{align}
It is essential to choose the inhomogeneities $x_1,\ldots, x_L$ in the two transfer matrices identically.
From the weight conservation property of $S^{k,1}(z)$ and the periodic boundary condition,
$T^k(z)$ preserves each sector $\mathbb{V}({\bf m})$ in \eqref{Vm}.

Let us examine the diagonal elements of $T^k(z)$ for general $k \in \{0,\ldots, n+1\}$.
When $\boldsymbol{\sigma}' = \boldsymbol{\sigma}$, 
all the arrays ${\bf a}_j$ in Figure \ref{fig:tk} become identical due to the weight conservation.
Thus, by employing \eqref{Srepeat0},  we obtain 
\begin{align}\label{sts}
\langle \boldsymbol{\sigma}|T^k(z)| \boldsymbol{\sigma}\rangle
= \sum_{{\bf a} \in \BB^k}\prod_{j=1}^L
S \Bigl( \frac{z}{x_j} \Bigr)^{{\bf a}, {\bf e}_{\sigma_j}}_{{\bf a}, {\bf e}_{\sigma_j}}
= \sum_{{\bf a} \in \BB^k}\prod_{j=1}^L t^{a_{1+\sigma_j}+\cdots + a_n}
\left( 1-t^{a_{\sigma_j}} \frac{z}{x_j} \right).
\end{align}

In the special cases of $k=0$ and $n+1$, one has $\BB^0=\{{\bf 0}:=(0,\ldots, 0)\}$ and 
$\BB^{n+1}=\{{\bf 1} := (1,\ldots, 1)\}$ in the multiplicity representation \eqref{Bk}.
Then, the RHS of \eqref{Srepeat0} becomes 
$[{\bf a}={\bf i}={\bf 0}, j=b](1-z)$ for $k=0$ and 
$[{\bf a}={\bf i} = {\bf 1}, j=b]t^{n-j}(1-tz)$ for $k=n+1$.
This implies that $T^0(z)$ and $T^{n+1}(z)$ are diagonal, with their elements obtained by 
reducing the sum \eqref{sts} to the terms ${\bf a}={\bf 0}$ and ${\bf 1}$, respectively. 
Consequently we have 
\begin{align}
T^0(z) &= \prod_{j=1}^L\left(1-\frac{z}{x_j}\right) \mathrm{Id},
\label{T0z}
\\
T^{n+1}(z) &= t^{K_1+\cdots + K_n}\prod_{j=1}^L\left(1-\frac{tz}{x_j}\right) \mathrm{Id} 
\label{Tn1}
\end{align}
on $\mathbb{V}({\bf m})$, where $K_i$ is defined in \eqref{Ki}.

Suppose that we are in the sector $T^k(z) \in \mathrm{End}(\mathbb{V}({\bf m}))$ \eqref{Vm}, hence 
$\boldsymbol{\sigma} =(\sigma_1,\ldots, \sigma_L) \in \mathcal{S}({\bf m})$ as in \eqref{Sm}.
From \eqref{sts}, we have
\begin{equation}
\begin{split}
\langle \boldsymbol{\sigma}|T^k(0)| \boldsymbol{\sigma}\rangle
&=  \sum_{{\bf a} \in \BB^k} t^{\sum_{j=1}^L(a_{1+\sigma_j}+\cdots + a_n)}
\\
&=  \sum_{{\bf a} \in \BB^k} t^{m_0(a_1+\cdots + a_n)  + m_1(a_2+\cdots + a_n) + \cdots + m_{n-1}a_n}
\\
&= \sum_{{\bf a} \in \BB^k} t^{m_0a_1+(m_0+m_1)a_2+\cdots + (m_0+\cdots + m_{n-1})a_n}.
\end{split}
\end{equation}
We write  the derivative simply as $\dot{T}^k(z) = \frac{dT^k(z)}{dz}$.
It is not diagonal, but the calculation of the diagonal elements goes similarly as 
\begin{align}
\langle \boldsymbol{\sigma}|\dot{T}^k(0)| \boldsymbol{\sigma}\rangle
= -\sum_{j=1}^L\frac{1}{x_j}
\sum_{{\bf a} \in \BB^k} t^{a_{\sigma_j}+m_0a_1+(m_0+m_1)a_2+\cdots + (m_0+\cdots + m_{n-1})a_n}.
\end{align}
These results are described as 
\begin{align}
\langle \boldsymbol{\sigma}|T^k(0)| \boldsymbol{\sigma}\rangle
&= e_k(t^{K_0},\ldots, t^{K_n}),
\qquad 
\label{T0}\\
\langle \boldsymbol{\sigma}|\dot{T}^k(0)| \boldsymbol{\sigma}\rangle
&= -\sum_{j=1}^L\frac{1}{x_j} e_k(u^{(\sigma_j)}_0, \ldots, u^{(\sigma_j)}_n),
\qquad 
u^{(\sigma)}_i = t^{\delta_{i, \sigma}+K_i} .
\label{Tp0}
\end{align}
where $K_i$ is defined in \eqref{Ki}.
In particular $u_0^{(\sigma)} = t^{\delta_{0,\sigma}}$. 
The functions $e_0,\ldots, e_{n+1}$ are \emph{elementary symmetric polynomials} in $(n+1)$ variables defined by
\begin{align} 
e_k(w_0,\ldots, w_n) = \sum_{{\bf a}\in \BB^k}w_0^{a_0}\cdots w_n^{a_n},
\label{ekb}
\end{align}
which satisfy the defining generating functional relation: 
\begin{align} 
&(1+\zeta w_0)\cdots (1+ \zeta w_n) = \sum_{k=0}^{n+1}\zeta^k e_k(w_0,\ldots, w_n).
\label{ekg}
\end{align}
We understand that $e_k(w_0,\ldots, w_n)=0$ for $k>n+1$.
A useful relation is 
\begin{align}\label{De}
\sum_{k=0}^{n+1}(-1)^ke_k(u^{(\sigma)}_0,\ldots, u^{(\sigma)}_n)
&=(1-u^{(\sigma)}_0)\cdots (1-u^{(\sigma)}_n) = \delta_{\sigma,0}D_{\bf m}
\qquad (0 \le \sigma \le n),
\end{align}
where $D_{\bf m}$ is defined in \eqref{Dm}.

From the definition \eqref{tkdef}--\eqref{tke} and \eqref{Srepeat0},  the transfer matrix 
$T^k(z)$ is diagonal at $z=0$. 
Therefore,  \eqref{T0} implies  
\begin{align}
T^k(0) = e_k(t^{K_0},\ldots, t^{K_n}) \mathrm{Id}\;\; \text{on } \mathbb{V}({\bf m}).
\label{tk0}
\end{align}

\begin{example}\label{ex:tk}
Following Example \ref{ex:H}, we set $n=2$ and $L=4$.
We denote the coefficient of the diagonal term generated by $T^k(z)$ by $\mathcal{D}_k(z)$.
\begin{equation}
\begin{split}
T^0(z)|0121\rangle &= \mathcal{D}_0(z) |0121\rangle,
\\
T^1(z)|0121\rangle &=\frac{(1-t)^4 z^2}{x_2 x_3} |1012\rangle 
- \frac{(1-t)^3 z (z - x_2)}{x_2 x_3} |1102\rangle \\
&+ \frac{(1-t)^2 t z (z - x_1)(t z - x_2)}{x_1 x_2 x_3} |0112\rangle 
+ \frac{(1-t)^2 z (z - x_2)(z - x_3)}{x_2 x_3 x_4} |1120\rangle \\
&+ \frac{(1-t)^2 t^2 z (z - x_1)(z - x_4)}{x_1 x_3 x_4} |0211\rangle
- \frac{(1-t)^3 t z^2 (z - x_4)}{x_2 x_3 x_4} |2011\rangle \\
&+ \frac{(1-t)^2 t z (z - x_2)(z - x_4)}{x_2 x_3 x_4} |2101\rangle 
+ \frac{(1-t)^2 z (z - x_3)(t z - x_4)}{x_2 x_3 x_4} |1021\rangle \\
& +\mathcal{D}_1(z) |0121\rangle,
\\
T^2(z)|0121\rangle &=\frac{(1-t)^2 t z (t z - x_1)(t z - x_2)}{x_1 x_2 x_3} |0112\rangle 
+ \frac{(1-t)^4 t^2 z^2}{x_3 x_4} |1210\rangle \\
&+\frac{(1-t)^3 t^2 z^2 (t z - x_2)}{x_2 x_3 x_4} |2110\rangle 
+ \frac{(1-t)^2 t^3 z (z - x_2)(t z - x_3)}{x_2 x_3 x_4} |1120\rangle \\
&+ \frac{(1-t)^2 t^2 z (t z - x_1)(z - x_4)}{x_1 x_3 x_4} |0211\rangle 
+ \frac{(1-t)^3 t^2 z (t z - x_4)}{x_3 x_4} |1201\rangle \\
&+ \frac{(1-t)^2 t^2 z (t z - x_2)(t z - x_4)}{x_2 x_3 x_4} |2101\rangle 
+ \frac{(1-t)^2 t^3 z (t z - x_3)(t z - x_4)}{x_2 x_3 x_4} |1021\rangle \\
&+ \mathcal{D}_2(z) |0121\rangle,
\\
T^3(z)|0121\rangle &= \mathcal{D}_3(z) |0121\rangle.
\end{split}
\end{equation}
The functions $\mathcal{D}_0(z)$ and $\mathcal{D}_3(z)$ are 
explicitly given by \eqref{T0z} and \eqref{Tn1} with $n=2$, respectively.
They lead to 
\begin{equation}\label{teg}
\begin{split}
\dot{T}^0(0)|0121\rangle &= \dot{\mathcal{D}}_0(0) |0121\rangle,
\\
\dot{T}^1(0)|0121\rangle &= \frac{(1-t)^2}{x_2} |1021\rangle 
+ \frac{(1-t)^2 t}{x_3} |0112\rangle 
+ \frac{(1-t)^2 t^2}{x_3} |0211\rangle \\
&+\frac{(1-t)^3}{x_3} |1102\rangle 
+ \frac{(1-t)^2 t}{x_3} |2101\rangle 
+ \frac{(1-t)^2}{x_4} |1120\rangle
+\dot{\mathcal{D}}_1(0) |0121\rangle,
\\
\dot{T}^2(0)|0121\rangle &=\frac{(1-t)^2 t^3}{x_2} |1021\rangle 
+ \frac{(1-t)^2 t}{x_3} |0112\rangle 
+ \frac{(1-t)^2 t^2}{x_3} |0211\rangle \\
&- \frac{(1-t)^3 t^2}{x_3} |1201\rangle 
+ \frac{(1-t)^2 t^2}{x_3} |2101\rangle 
+ \frac{(1-t)^2 t^3}{x_4} |1120\rangle
+\dot{\mathcal{D}}_2(0) |0121\rangle,
\\
\dot{T}^3(0)|0121\rangle &= \dot{\mathcal{D}}_3(0) |0121\rangle,
\end{split}
\end{equation}
where $\dot{\mathcal{D}}_k(0)= \left.\frac{d\mathcal{D}_k(z)}{dz}\right|_{z=0}$
is available from \eqref{Tp0}.
\end{example}

\section{$\Hp$  from transfer matrices}\label{sec:H}

Let us introduce a linear combination of the special value of the differentiated  transfer matrices as
\begin{align}\label{cH}
\mathcal{H}=D^{-1}_{\bf m}
\sum_{k=0}^{n+1}(-1)^{k-1}\dot{T}^k(0) -\left( \sum_{j=1}^L\frac{1}{x_j} \right)\mathrm{Id},
\end{align}
where $D_{\bf m}$ is given in \eqref{Dm}.
It defines a linear operator on each sector $\mathbb{V}({\bf m})$.

The main result of this paper is the following.
\begin{theorem}\label{th:main}
The Markov matrix $\Hp$ of the $t$-PushTASEP in \eqref{Hdef}--\eqref{wdef} 
is identified with $\mathcal{H}$ \eqref{cH} based on the transfer matrices in  Section \ref{sec:T}. 
Namely the following equality holds in each sector $\mathbb{V}({\bf m})$:
\begin{align}\label{HH}
\Hp = \mathcal{H}.
\end{align}
\end{theorem}

\begin{example}\label{ex:cH}
Set $n=2, L=4$ following  Example \ref{ex:H} and  Example \ref{ex:tk}. 
From \eqref{teg} we have 
\begin{equation}\label{exv}
\begin{split}
\sum_{k=0}^{3}(-1)^{k-1}\dot{T}^k(0)|0121\rangle
&=\frac{(1-t)^3 (1 + t + t^2)}{x_2} |1021\rangle 
+\frac{(1-t)^3}{x_3} |1102\rangle 
+\frac{(1-t)^3 t^2}{x_3} |1201\rangle \\
&+\frac{(1-t)^3 t}{x_3} |2101\rangle 
+\frac{(1-t)^3 (1 + t + t^2)}{x_4} |1120\rangle + 
\dot{\mathcal{D}}(0) |0121\rangle,
\end{split}
\end{equation}
where $\dot{\mathcal{D}}(0)= \sum_{k=0}^3(-1)^{k-1}\dot{\mathcal{D}}_k(0)$.
(See Example \ref{ex:tk} for the definition of $\mathcal{D}_k(z)$.)
These vectors belong to the sector $\mathbb{V}({\bf m})$ 
with multiplicity ${\bf m} = (1,2,1)$.
Thus we have $D_{\bf m} = (1-t)^2(1-t^3)$ according to \eqref{Dm}.
The vector \eqref{exv} divided by $D_{\bf m}$ reproduces  Example \ref{ex:H},
where the coincidence of the diagonal terms will be shown in \eqref{diag1}.
\end{example}

The rest of this section is devoted to the proof of Theorem \ref{th:main}.

\subsection{Diagonal elements}
As a warm-up,  we first prove \eqref{HH}  for the diagonal matrix elements, i.e.,
\begin{align}
\langle \boldsymbol{\sigma}|\Hp| \boldsymbol{\sigma}\rangle
=\langle \boldsymbol{\sigma}|\mathcal{H}| \boldsymbol{\sigma}\rangle.
\end{align}
From \eqref{Hdef} we know 
$\langle \boldsymbol{\sigma}|\Hp| \boldsymbol{\sigma}\rangle = 
-\sum_{j=1}^L\frac{[\sigma_j\ge 1]}{x_j}$.
The RHS is calculated as
\begin{align}
\langle \boldsymbol{\sigma}|\mathcal{H}| \boldsymbol{\sigma}\rangle
&= D_{\bf m}^{-1}\sum_{k=0}^{n+1}(-1)^{k-1}
\langle \boldsymbol{\sigma}|\dot{T}^k(0)| \boldsymbol{\sigma}\rangle -\sum_{j=1}^L\frac{1}{x_j}
\nonumber\\
&\overset{\eqref{Tp0}}{=}
D_{\bf m}^{-1}\sum_{j=1}^L\frac{1}{x_j}
\sum_{k=0}^{n+1}(-1)^k e_k(u^{(\sigma_j)}_0, \ldots, u^{(\sigma_j)}_n)-\sum_{j=1}^L\frac{1}{x_j}
\nonumber\\
&\overset{\eqref{De}}{=}\sum_{j=1}^L \frac{\delta_{\sigma_j,0}}{x_j}
 -\sum_{j=1}^L\frac{1}{x_j}
= -\sum_{j=1}^L\frac{[\sigma_j\ge 1]}{x_j},
\label{diag1}
\end{align}
which matches $\langle \boldsymbol{\sigma}|\Hp| \boldsymbol{\sigma}\rangle$ as required.

\subsection{Reduced diagram and its depth}
From now on, we assume $\boldsymbol{\sigma}'\neq \boldsymbol{\sigma}$
and concentrate on the off-diagonal elements 
$\langle \boldsymbol{\sigma}'|\Hp| \boldsymbol{\sigma}\rangle$ and 
$\langle \boldsymbol{\sigma}'|\mathcal{H}| \boldsymbol{\sigma}\rangle$.
The former is given, from \eqref{Hdef}, as
\begin{subequations}
\begin{align}
\langle \boldsymbol{\sigma}'|\Hp| \boldsymbol{\sigma}\rangle
&= \sum_{j=1}^L 
\langle \boldsymbol{\sigma}'|\Hp| \boldsymbol{\sigma}\rangle_j,
\quad {\text{with}}
\label{Hx}
\\
\langle \boldsymbol{\sigma}'|\Hp| \boldsymbol{\sigma}\rangle_j
&= \frac{1}{x_j} \prod_{1 \le h \le n}
w^{(j)}_{\boldsymbol{\sigma}, \boldsymbol{\sigma}'}(h),
\label{Hj}
\end{align}
\end{subequations}
where the factor $w^{(j)}_{\boldsymbol{\sigma}, \boldsymbol{\sigma}'}(h)$ 
has been defined in \eqref{wdef}.
On the other hand $\langle \boldsymbol{\sigma}'|\mathcal{H}| \boldsymbol{\sigma}\rangle$ is given,
from \eqref{tke} and \eqref{cH}, as
\begin{subequations}
\begin{align}
\langle \boldsymbol{\sigma}'|\mathcal{H} |\boldsymbol{\sigma}\rangle
&=  D_{\bf m}^{-1}\sum_{k=0}^{n+1}(-1)^{k-1}
\sum_{j=1}^L\langle \boldsymbol{\sigma}'|\dot{T}^k(0)| \boldsymbol{\sigma}\rangle_j
\quad {\text{with}}
\label{Hct}
\\
\langle \boldsymbol{\sigma}'|\dot{T}^k(0)| \boldsymbol{\sigma}\rangle_j
&= \frac{1}{x_j} \sum_{{\bf a}_1,\ldots, {\bf a}_L \in \BB^k}
S(0)^{{\bf a}_2, {\bf e}_{\sigma'_1}}_{{\bf a}_1, {\bf e}_{\sigma_1}}
\cdots 
\dot{S}(0)^{{\bf a}_{j+1}, {\bf e}_{\sigma'_j}}_{{\bf a}_j, {\bf e}_{\sigma_j}} 
\cdots 
S(0)^{{\bf a}_1, {\bf e}_{\sigma'_L}}_{{\bf a}_L, {\bf e}_{\sigma_L}}.
\label{Hcj}
\end{align}
\end{subequations}
where $\dot{S}(z)= \frac{dS(z)}{dz}$.
Thus the equality 
$\langle \boldsymbol{\sigma}'|\Hp| \boldsymbol{\sigma}\rangle= 
\langle \boldsymbol{\sigma}'|\mathcal{H}| \boldsymbol{\sigma}\rangle$
for any $\boldsymbol{\sigma}\neq \boldsymbol{\sigma}' \in S({\bf m})$
follows once we show
\begin{align}
\prod_{1 \le h \le n}
w^{(j)}_{\boldsymbol{\sigma}, \boldsymbol{\sigma}'}(h)
= D_{\bf m}^{-1}\sum_{k=0}^{n+1}(-1)^{k-1}
x_j \langle \boldsymbol{\sigma}'|\dot{T}^k(0)| \boldsymbol{\sigma}\rangle_j.
\label{HHj}
\end{align}
This relation already achieves  two simplifications from the original problem.
Specifically, there is no summation over the sites $j=1,\ldots, L$, and the dependence on 
$x_1,\ldots, x_L$ is eliminated, leaving it dependent only on the parameter $t$.
We list the necessary data for  $S(0)$ and $\dot{S}(0)$ in Table \ref{tab:S}.

\begin{table}[h]
\centering 
\caption{Special values $S(0)^{{\bf a}, {\bf e}_b}_{{\bf i}, {\bf e}_c}$ and 
$\dot{S}(0)^{{\bf a}, {\bf e}_b}_{{\bf i}, {\bf e}_c}$ obtained from \eqref{Srepeat0} relevant to $\dot{T}^k(0)$.
The symbols $\delta$ and $\varepsilon$ are shorthand for 
$\delta=\delta^{{\bf a}+{\bf e}_b}_{{\bf i} + {\bf e}_c}$ and 
$\varepsilon = (-1)^{a_0+\cdots + a_{c-1}+i_0+\cdots + i_{b-1}}$, respectively.
For the nonzero cases with $c \neq b$,  we use the fact $a_c=1$ which follows from
the constraint ${\bf a}+{\bf e}_b = {\bf i} + {\bf e}_c$.
Similarly, the sign factor for the $c=b$ case has been set to $\varepsilon=1$.
The second line with $c=b$ case is found to be irrelevant and is therefore omitted.
}
\label{tab:S} 

\vspace*{0.5cm}
\begin{tabular}{c|c|c|c}
\hline 
& $c<b$  &  $c=b$  & $c>b$ \\ 
\hline
$S(0)^{{\bf a}, {\bf e}_b}_{{\bf i}, {\bf e}_c}$ 
 &   $\delta \varepsilon t^{a_{c+1}+\cdots + a_n}(1-t)$ 
 &  $\delta  t^{a_{c+1}+\cdots + a_n}$  
 &  0  
 \\ 
$\dot{S}(0)^{{\bf a}, {\bf e}_b}_{{\bf i}, {\bf e}_c}$ 
 &   0 
 &  $---$
 &  $\delta \varepsilon t^{a_{c+1}+\cdots + a_n}(1-t)$ 
 \\ 
 \hline
 \end{tabular}
\end{table}

We depict $x_j\langle \boldsymbol{\sigma}'|\dot{T}^k(0)| \boldsymbol{\sigma}\rangle_j$
as in Figure \ref{fig:tk},  suppressing all the spectral parameters $z/x_i$ as they are set to zero.
All the vertical arrows from $\sigma_i$ to $\sigma'_i$ with $\sigma_i = \sigma'_i$, 
corresponding to ``diagonal transitions", are omitted.
Moreover, we perform a cyclic shift such that the site $j$ 
appears in the leftmost position (this is merely for ease of visualization and not essential), 
attaching it with $\circ$ to indicate that 
$\dot{S}(0)$ should be applied there, in contrast to $S(0)$ for other sites.  
Such a diagram will be referred to as {\em reduced diagram}.
See \eqref{red}, where ${\bf a}_i \in \BB^k$, $s_i\neq r_i \in \{0,\ldots, n\}$ for 
$0 \le i \le g$ with some $1 \le g <L$.

\begin{equation}
\begin{tikzpicture}

\node at (-1.8,-0.18){$\displaystyle{\sum_{\phantom{AA}{\bf a}_0,\ldots, {\bf a}_g \in \BB^k}}$};
  % vertical 
  \draw[->] (0.6,-0.6) node[below]{$r_0$}-- (0.6,0.6) node[above]{$s_0$};
  \draw (0.59, 0) circle[radius=0.1cm];
  % horizontal
  \draw[->,line width=0.6mm] (0,0) node[left]{${\bf a}_0$}-- (1.2,0) node[right]{${\bf a}_1$};
\begin{scope}[shift={(1.8,0)}]
  \draw[->] (0.6,-0.6) node[below]{$r_1$}-- (0.6,0.6) node[above]{$s_1$};
  % horizontal
  \draw[->,line width=0.6mm] (0,0) -- (1.2,0) node[right]{${\bf a}_2$};
   \end{scope}

  \draw[->,line width=0.6mm] (3.6,0) -- (4.4,0);
  \node at (4.8,0) {$\cdots$};
  \draw[->,line width=0.6mm] (5.1,0) -- (5.9,0);
   
   \begin{scope}[shift={(6.6,0)}]
     \draw[->] (0.6,-0.6) node[below]{$r_g$}-- (0.6,0.6) node[above]{$s_g$};
  % horizontal
  \draw[->,line width=0.6mm] (0,0) node[left]{${\bf a}_g$}-- (1.2,0) node[right]{${\bf a}_0$};
     \end{scope}
\end{tikzpicture}
\label{red}
\end{equation}
The diagram should be understood as representing the sum in \eqref{Hcj}, where the $L-g-1$ vertical arrows 
corresponding to the diagonal transitions are suppressed, but their associated vertex weights 
should still be accounted for.
Since the carriers ${\bf a}_i$'s remain unchanged when crossing the omitted vertical arrows,
the summation reduces to those over ${\bf a}_0, \ldots, {\bf a}_g$, where 
${\bf a}_{i+1} = {\bf a}_i + {\bf e}_{r_i}-{\bf e}_{s_i}$ $(i \mod g+1)$. 

\begin{lemma}\label{le:rs}
$\langle \boldsymbol{\sigma}'|\dot{T}^k(0)| \boldsymbol{\sigma}\rangle_j = 0$,
unless the reduced diagram \eqref{red} for it satisfies the conditions
\begin{subequations}
\begin{align}
&\{r_0,\ldots, r_g\} = \{s_0, \ldots, s_g\} = 
\{h_0, \ldots, h_g\},
\label{rscon1}\\
&(r_0, s_0) = (h_g, h_0), \quad 
 (r_i,s_i) = (h_{q_i}, h_{q_i+1}) 
\label{rscon2}
\end{align}
for some sequence $0 \le  h_0 < \cdots < h_g \le n$ and 
$0 \le q_i \le g-1$\;$ (i=1,\ldots, g)$.
\end{subequations}
\end{lemma}
\begin{proof}
From weight conservation, \eqref{red} vanishes unless the condition (i) 
$\{r_0,\ldots, r_g\} = \{s_0, \ldots, s_g\}$ holds as multisets.
From Table \ref{tab:S}, it also vanishes unless the additional conditions (ii) 
$r_0>s_0$, $r_1<s_1,\ldots, r_g<s_g$ are satisfied.
Conditions (i) and (ii) together are equivalent to \eqref{rscon1} and \eqref{rscon2}.
\end{proof}

\begin{example}
A reduced diagram \eqref{red} corresponding to 
$(h_0,\ldots, h_3)=(0,2,3,4)$ with $g=3$ and 
$(r_0,s_0), \ldots, \allowbreak (r_3,s_3) =(4,0),(2,3), (0,2), (3,4)$ is shown in (b) in Example \ref{ex:d}.
It is associated with the following motion of particles.
\begin{equation}
\label{color2}
\vcenter{\hbox{%
\begin{tikzpicture}[scale=1, line cap=round, line join=round, thick, every node/.style={scale=1}]
  \def\xA{-1.92} \def\xB{-0.64} \def\xC{0.64} \def\xD{1.92}
  \def\ymax{0.525} \def\ymin{-0.525} \def\offset{0.10}
  \tikzset{every path/.style={line width=1.8pt, rounded corners=5pt}}

  % --- Black (0)
  \draw[->, black]
      (\xA-0.64, 0+\offset) -- (\xA, 0+\offset) -- (\xA, \ymax)
      node[above,black]{0};
  \draw[->, black]
      (\xC, \ymin) node[below,black]{0} -- (\xC, -\offset) -- (\xD+0.64, -\offset);

  % --- Blue (3)
  \draw[->, blue]
      (\xA-0.64, 0-\offset) -- (\xB, 0-\offset) -- (\xB, \ymax)
      node[above,blue]{3};
  \draw[->, blue]
      (\xD, \ymin) node[below,blue]{3} -- (\xD, -\offset) -- (\xD+0.64, -\offset);

  % --- Red (2)
  \draw[->, red]
      (\xB, \ymin) node[below,red]{2} -- (\xB, -\offset) --
      (\xC, -\offset) -- (\xC, \ymax) node[above,red]{2};

  % --- Green (4)
  \draw[->, green!60!black]
      (\xA, \ymin) node[below,green!60!black]{4} -- (\xA, \offset) --
      (\xD, \offset) -- (\xD, \ymax) node[above,green!60!black]{4};
\end{tikzpicture}%
}}
\end{equation}
\end{example}

The increasing sequence $(h_0,\ldots, h_g)$ appearing in  Lemma \ref{le:rs} 
represents the list of particle types moved during the transition 
$\boldsymbol{\sigma} \rightarrow \boldsymbol{\sigma}'$ induced by $\dot{T}^k(0)$.
We refer to this sequence  as the {\em moved particle types}.
By the definition, 
$$g= (\text{number of moved particle types}) -1.$$

Suppose the diagram \eqref{red} satisfies \eqref{rscon1} and \eqref{rscon2}
for some moved particle types.
To ensure weight conservation at every vertex, the capacity $k$ 
of the carriers must be at least a certain value.
We define the minimum possible capacity as the {\em depth} $d$ of the reduced diagram or the transition 
$\boldsymbol{\sigma} \rightarrow \boldsymbol{\sigma}'$.
Clearly, the depth is unaffected by the diagonal part of the transition 
which is suppressed in the reduced diagram.
We refer to the carries whose capacity equals the depth as {\em minimal carries}. 

\begin{example}\label{ex:d}
Reduced diagrams and the minimal carries corresponding to the moved particle types 
(a) $(1,2,4)$ and (b), (c) $(0,2,3,4)$.  The depth $d$ of (a), (b) and (c) are 1, 2 and 3, respectively.
\begin{align}
&\begin{tikzpicture}[scale=0.7]
\node at (3.14,1.8){(a)\, $d=1$};
  % vertical 
  \draw[->] (0.6,-0.6) node[below]{$4$}-- (0.6,0.6) node[above]{$1$};
  \draw (0.59, 0) circle[radius=0.1cm];
  % horizontal
  \node at   (-0.3,0) {$\scriptstyle{1}$}; 
  \draw[->,line width=0.6mm] (0,0) node[left]{}-- (1.2,0) node[right]{};
\begin{scope}[shift={(1.8,0)}]
  \draw[->] (0.6,-0.6) node[below]{$2$}-- (0.6,0.6) node[above]{$4$};
  % horizontal
 \node at   (-0.3,0) {$\scriptstyle{4}$}; 
  \draw[->,line width=0.6mm] (0,0) -- (1.2,0) node[right]{};
 \end{scope}
\begin{scope}[shift={(3.6,0)}]
  \draw[->] (0.6,-0.6) node[below]{$1$}-- (0.6,0.6) node[above]{$2$};
  % horizontal
   \node at   (-0.3,0) {$\scriptstyle{2}$}; 
  \draw[->,line width=0.6mm] (0,0) -- (1.2,0) node[right]{};
     \node at   (1.4,0) {$\scriptstyle{1}$}; 
 \end{scope}
\end{tikzpicture}
\quad\;
\begin{tikzpicture}[scale=0.7]
\node at (4,1.8){(b)\, $d=2$};
  % vertical 
  \draw[->] (0.6,-0.6) node[below]{$4$}-- (0.6,0.6) node[above]{$0$};
  \draw (0.59, 0) circle[radius=0.1cm];
  % horizontal
  \node at   (-0.3,0.22) {$\scriptstyle{0}$}; \node at   (-0.3,-0.22) {$\scriptstyle{3}$};
  \draw[->,line width=0.6mm] (0,0) node[left]{}-- (1.2,0) node[right]{};
\begin{scope}[shift={(1.8,0)}]
  \draw[->] (0.6,-0.6) node[below]{$2$}-- (0.6,0.6) node[above]{$3$};
  % horizontal
 \node at   (-0.3,0.22) {$\scriptstyle{3}$}; \node at   (-0.3,-0.22) {$\scriptstyle{4}$};
  \draw[->,line width=0.6mm] (0,0) -- (1.2,0) node[right]{};
 \end{scope}
\begin{scope}[shift={(3.6,0)}]
  \draw[->] (0.6,-0.6) node[below]{$0$}-- (0.6,0.6) node[above]{$2$};
  % horizontal
   \node at   (-0.3,0.22) {$\scriptstyle{2}$}; \node at   (-0.3,-0.22) {$\scriptstyle{4}$};
  \draw[->,line width=0.6mm] (0,0) -- (1.2,0) node[right]{};
 \end{scope}
\begin{scope}[shift={(5.4,0)}]
  \draw[->] (0.6,-0.6) node[below]{$3$}-- (0.6,0.6) node[above]{$4$};
  % horizontal
   \node at   (-0.3,0.22) {$\scriptstyle{0}$}; \node at   (-0.3,-0.22) {$\scriptstyle{4}$};
  \draw[->,line width=0.6mm] (0,0) -- (1.2,0) node[right]{};
     \node at   (1.5,0.2) {$\scriptstyle{0}$}; \node at   (1.5,-0.2) {$\scriptstyle{3}$};
 \end{scope}
\end{tikzpicture}
\cr
& \begin{tikzpicture}[scale=0.7]
\node at (4,1.8){(c)\, $d=3$};
  % vertical 
  \draw[->] (0.6,-0.6) node[below]{$4$}-- (0.6,0.6) node[above]{$0$};
  \draw (0.59, 0) circle[radius=0.1cm];
  % horizontal
\node at   (-0.34,0.43) {$\scriptstyle{0}$}; \node at   (-0.34,-0.42) {$\scriptstyle{3}$};
  \draw[->,line width=0.6mm] (-0.04,0) node[left]{$\scriptstyle 2$}-- (1.2,0) node[right]{$\scriptstyle 3$};
\begin{scope}[shift={(1.8,0)}]
  \draw[->] (0.6,-0.6) node[below]{$0$}-- (0.6,0.6) node[above]{$2$};
  % horizontal
  \node at   (-0.28,0.43) {$\scriptstyle{2}$}; \node at   (-0.28,-0.42) {$\scriptstyle{4}$};
  \draw[->,line width=0.6mm] (0,0) -- (1.2,0) node[right]{$\scriptstyle 3$};
 \end{scope}
\begin{scope}[shift={(3.6,0)}]
  \draw[->] (0.6,-0.6) node[below]{$2$}-- (0.6,0.6) node[above]{$3$};
  % horizontal
  \node at   (-0.28,0.43) {$\scriptstyle{0}$}; \node at   (-0.28,-0.42) {$\scriptstyle{4}$};
  \draw[->,line width=0.6mm] (0,0) -- (1.2,0) node[right]{$\scriptstyle 2$};
 \end{scope}
\begin{scope}[shift={(5.4,0)}]
  \draw[->] (0.6,-0.6) node[below]{$3$}-- (0.6,0.6) node[above]{$4$};
  % horizontal
  \node at   (-0.28,0.43) {$\scriptstyle{0}$}; \node at   (-0.27,-0.42) {$\scriptstyle{4}$};
  \draw[->,line width=0.6mm] (0,0) -- (1.2,0) node[right]{$\scriptstyle 2$};
  \node at   (1.54,0.43) {$\scriptstyle{0}$}; \node at   (1.54,-0.42) {$\scriptstyle{3}$};
 \end{scope}
\end{tikzpicture}
\label{dab}
\end{align}

\noindent
Here we have employed the tableau representation \eqref{tk} for the carriers.
The comparison between (b) and (c)  demonstrates that the depth depends on the {\em ordering} of the 
vertical arrows $s_i \rightarrow r_i$, even when they correspond to the same 
moved particle types.
\end{example}

Example \ref{ex:d} also demonstrates that $d\le g$ in general, and 
the union of tableau letters contained in the minimal carriers 
${\bf a}_0, \ldots, {\bf a}_g$
coincide with the moved particle types
$\{h_0,\ldots, h_g\}$ as sets. 
Moreover, they are uniquely determined from $\boldsymbol{\sigma}$ and $\boldsymbol{\sigma}'$,
reducing the sum \eqref{red} into a {\em single} term.
In fact, in the reduced diagram \eqref{red}, ${\bf a}_0,\ldots, {\bf a}_g \in \BB^{k=d}$ are determined 
by the recursion relation ${\bf a}_{i+1} = {\bf a}_i + {\bf e}_{r_i}-{\bf e}_{s_i}$ $(i \mod g+1)$ and 
the ``initial condition":
\begin{align}
{\bf a}_0 &=\{s_0\}  \cup \mathscr{S}_1 \cup \cdots \cup \mathscr{S}_g,
\qquad
\mathscr{S}_i = \begin{cases} 
\varnothing & \text{if }\; s_i \in \{r_0,\dots, r_{i-1}\},\\
\{s_i\} & \text{otherwise}.
\end{cases}
\end{align}
To summarize the argument thus far, we have reduced the equality \eqref{HHj} slightly to 
\begin{align}
\prod_{1 \le h \le n}
w^{(j)}_{\boldsymbol{\sigma}, \boldsymbol{\sigma}'}(h)
= D_{\bf m}^{-1}\sum_{k=d}^{n+1}(-1)^{k-1}
x_j \langle \boldsymbol{\sigma}'|\dot{T}^k(0)| \boldsymbol{\sigma}\rangle_j,
\label{HHj2}
\end{align}
where the lower bound of the sum over $k$ has been increased to the depth $d$ of the transition 
$\boldsymbol{\sigma} \rightarrow \boldsymbol{\sigma}'$.
The LHS is either zero or a nonzero rational function of $t$, whereas the RHS involves 
summations over $k$ as well as over carriers from $\BB^k$ entering the definition of 
$x_j \langle \boldsymbol{\sigma}'|\dot{T}^k(0)| \boldsymbol{\sigma}\rangle_j$ in \eqref{red}.

In the following, we divide the proof of \eqref{HHj2} into two cases, depending on whether  
its LHS is nonzero or zero.
The RHS in these corresponding situations will be referred to as 
{\em wanted terms} and {\em unwanted terms}, respectively.
From the definition of $w^{(j)}_{\boldsymbol{\sigma}, \boldsymbol{\sigma}'}(h)$
in Section \ref{sec:tpush}, unwanted terms correspond to the situation $s_0 \neq 0$.
In  Example \ref{ex:d}, (a) is unwanted while (b) and (c) are wanted.

\subsection{Wanted terms}\label{ss:want}

This subsection and the next form the technical focus of the proof.
From the definition of $w^{(j)}_{\boldsymbol{\sigma}, \boldsymbol{\sigma}'}(h)$ around \eqref{wdef}, 
the wanted terms generally correspond to the situation where the minimum $h_0$ of the 
moved particle types $(h_0,\ldots, h_g)$ in Lemma \ref{le:rs} is zero, i.e.,  $h_0=0$.
Then  \eqref{HHj2} is written down explicitly as
\begin{align}
\prod_{i=1}^g\frac{(1-t)t^{\ell_{h_i}}}{1-t^{K_{h_i}}}
= D_{\bf m}^{-1}\sum_{k=d}^{n+1}(-1)^{k-1}
x_j \langle \boldsymbol{\sigma}'|\dot{T}^k(0)| \boldsymbol{\sigma}\rangle_j.
\label{HHj3}
\end{align}
Our calculation of the RHS of \eqref{HHj3} consists of two steps.

{\em Step 1}. We consider the ``leading term" $k=d$ in the RHS of \eqref{HHj2}  and the 
corresponding reduced diagram, in which the carriers are uniquely determined, as shown in 
Example \ref{ex:d} (b) and (c).
We claim that
\begin{align}\label{kakyo}
x_j \langle \boldsymbol{\sigma}'|\dot{T}^d(0)| \boldsymbol{\sigma}\rangle_j
=(-1)^{d-1}(1-t) \prod_{i=1}^g(1-t)t^{\ell_{h_i}},
\end{align}
where $\ell_h$ has been defined prior to \eqref{wdef}.
Let us justify the origin of the constituent factors (i) sign, (ii) powers of $(1-t)$, (iii) powers of $t$, individually.

(i) The sign of a vertex can become negative only for non-diagonal transitions, which occur at the $g+1$ vertices 
in the reduced diagram \eqref{red}. 
From the comment following \eqref{Srepeat0} and the conditions in Lemma \ref{le:rs},
the $g+1$ vertices corresponding to the vertical arrows $r_i \rightarrow s_i$ in \eqref{rscon2}  
have $+$ signs for  $i=1,\ldots, g$ and
$(-1)^{d-1}$ for $i=0$.

(ii) From Table \ref{tab:S}, the contributions of $(1-t)$ at each of the $g+1$ vertices results in a factor of $(1-t)^{g+1}$.
 
(iii) From Table \ref{tab:S}, the power of $t$ can be evaluated as the sum of the quantities of the form  
$a_{c+1}+\cdots + a_n$ in the multiplicity representation of the carriers ${\bf a}=(a_0,\ldots, a_n)$, 
attached to each vertex.
This formula implies that a particle of type $h$ in the carriers 
contributes $[c<h] \in \{0, 1\}$ whenever it passes over a site $i$ occupied with $\sigma_i=c$.
Alternatively, this can be calculated as the total contribution collected by the moved particles 
$h_0,\ldots, h_g$ from the smaller-species particles in $\boldsymbol{\sigma}$.
This precisely leads to $\ell_{h_1}+\cdots + \ell_{h_g}$, where $\ell_{h_0}$ can be excluded due to $\ell_{h_0}=\ell_0=0$.
Thus the factor $\prod_{i=1}^g t^{\ell_{h_i}}$ is obtained as claimed.
The reformulation in the calculation described here is analogous to the transition from the Eulerian picture, 
which tracks properties at fixed spatial points, 
to the Lagrangian picture, which follows individual particles, in fluid mechanics.
In our context, it also incorporates the contribution from the vertices 
corresponding to the diagonal transitions efficiently via the quantities $\ell_h$'s.

{\em Step 2}. Let us turn to the $k=d+1,\ldots, n+1$ terms in \eqref{HHj2}.
We illustrate the idea of evaluating them along Example \ref{ex:d} (b) for $k=5$ and $n=7$
$(d=2, g=3)$.
The carriers from $\BB^5$ are no longer unique.
However, those satisfying the weight conservation with $\boldsymbol{\sigma}$ and $\boldsymbol{\sigma}'$
are exactly those obtained just by supplementing the common three letters  from the yet unused ones $\{1,5,6,7\}$
to the existing ones everywhere.
For instance, choosing them to be $1,5,6$, the carriers read 
$(0\underline{1}3\underline{56}), (\underline{1}34\underline{56}), 
(\underline{1}24\underline{56}), (0\underline{1}4\underline{56}), 
(0\underline{1}3\underline{56})$ from the left to the right, where the underlines signify 
the added letters.
Suppose the added letters are $\alpha, \beta, \gamma$.
Then, in the Lagrangian picture mentioned in the above, 
the effect of the supplement is to endow the RHS of \eqref{kakyo} with an extra factor 
$+\, t^{f_1m_0 + f_5(m_0+\cdots + m_4) + f_6(m_0 + \cdots + m_5) + f_7(m_0+\cdots + m_6)}$,
where 
$f_\lambda= [\lambda \in \{\alpha, \beta, \gamma\}]=0,1$ and 
$f_1+f_5+f_6+f_7=3$ reflecting that there are three letters to be added.
The sign factor is $+$ because a possible $-$ from any vertex with vertical arrow 
$h_{q_i} \rightarrow h_{q_i+1}$ is compensated by the leftmost vertex with vertical arrow 
$h_g \rightarrow 0$.
Now, the sum over non-unique carriers for $\dot{T}^5(0)$ 
becomes a sum over the ways to supplement extra letters to the minimal carriers.
Consequently we get 
\begin{equation}
\begin{split}
x_j \langle \boldsymbol{\sigma}'|\dot{T}^5(0)| \boldsymbol{\sigma}\rangle_j
=& x_j \langle \boldsymbol{\sigma}'|\dot{T}^2(0)| \boldsymbol{\sigma}\rangle_j
\\
& \times \!\sum_{\substack{f_1,f_5, f_6, f_7=0,1 \\ f_1+f_5+f_6+f_7=3} }
t^{f_1m_0 + f_5(m_0+\cdots + m_4) + f_6(m_0 + \cdots + m_5) + f_7(m_0+\cdots + m_6)}
\\
=& x_j \langle \boldsymbol{\sigma}'|\dot{T}^2(0)| \boldsymbol{\sigma}\rangle_j
e_3(t^{K_1},t^{K_5},t^{K_6},t^{K_7}),
\end{split}
\end{equation}
where $K_i$ is defined in \eqref{Ki} and $e_3$ is an elementary symmetric polynomial \eqref{ekb}.
In general, a similar argument leads to 
\begin{equation}\label{kakyo2}
x_j \langle \boldsymbol{\sigma}'|\dot{T}^k(0)| \boldsymbol{\sigma}\rangle_j
= x_j \langle \boldsymbol{\sigma}'|\dot{T}^d(0)| \boldsymbol{\sigma}\rangle_j
e_{k-d}(t^{K_{\bar{h}_1}}, \ldots, t^{K_{\bar{h}_{n-g}}})
\qquad (d \le k \le n+1),
\end{equation}
where $1 \le \bar{h}_1, \ldots, \bar{h}_{n-g} \le n$ are the types of unmoved particles specified as the complement:
\begin{align}\label{hbar}
\{0,\dots, n\} = \{h_0(=0), h_1,\ldots, h_g\}\sqcup \{\bar{h}_1, \ldots, \bar{h}_{n-g}\}.
\end{align}
Substituting \eqref{kakyo} and \eqref{kakyo2} into the RHS of \eqref{HHj3}
and using \eqref{Dm}, \eqref{hbar} and \eqref{ekg},  we obtain
\begin{equation}
\label{esum}
\begin{split}
\sum_{k=d}^{n+1} \frac{(-1)^{k-1}}{D_{\bf m}}
x_j \langle \boldsymbol{\sigma}'|\dot{T}^k(0)| \boldsymbol{\sigma}\rangle_j
&= \frac{(1-t)}{D_{\bf m}} \prod_{i=1}^g(1-t)t^{\ell_{h_i}}
\sum_{k=d}^{n+1}(-1)^{k-d}e_{k-d}(t^{K_{\bar{h}_1}}, \ldots, t^{K_{\bar{h}_{n-g}}})
\\
&=\frac{(1-t) \prod_{i=1}^g(1-t)t^{\ell_{h_i}}\prod_{i=1}^{n-g}(1-t^{K_{\bar{h}_i}})}
{(1-t) \prod_{i=1}^g(1-t^{K_{h_i}}) \prod_{i=1}^{n-g}(1-t^{K_{\bar{h}_i}})}
= \prod_{i=1}^g\frac{(1-t)t^{\ell_{h_i}}}{1-t^{K_{h_i}}},
\end{split}
\end{equation}
completing the proof of \eqref{HHj3}.

\subsection{Unwanted terms}\label{ss:unwant}
 
 The unwanted terms correspond to the case where the minimum $h_0$ of the moved particle types 
 $(h_0,\ldots, h_g)$ in Lemma \ref{le:rs} is nonzero.
Thus we are to show 
\begin{align}\label{HHj4}
0 = \sum_{k=d}^{n+1}(-1)^{k-1}
x_j \langle \boldsymbol{\sigma}'|\dot{T}^k(0)| \boldsymbol{\sigma}\rangle_j
\end{align}
assuming that the reduced diagram of 
$x_j \langle \boldsymbol{\sigma}'|\dot{T}^k(0)| \boldsymbol{\sigma}\rangle_j$
has the form \eqref{red},  where $r_i$ and $s_i$ satisfy the conditions 
\eqref{rscon1} and \eqref{rscon2}  with $h_0 \in \{1,\ldots, n\}$.
All the arguments concerning the wanted terms persist until \eqref{kakyo2}.
A key difference arises at \eqref{hbar}, where $h_0 \neq 0$ results in
$0 \in \{\bar{h}_1, \ldots, \bar{h}_{n-g}\}$.
Since $K_0=0$, the summation 
 $\sum_{k=d}^{n+1}(-1)^{k-d}e_{k-d}(t^{K_{\bar{h}_1}}, \ldots, t^{K_{\bar{h}_{n-g}}})
 =\prod_{i=1}^{n-g}(1-t^{K_{\bar{h}_i}})$ 
 involved in \eqref{esum} vanishes.
 
 We note that in the above calculation and \eqref{esum},
 the summand $e_{k-d}(t^{K_{\bar{h}_1}}, \ldots, t^{K_{\bar{h}_{n-g}}})$ is actually zero for $k=n+1$,
as the index $n+1-d$ exceeds the number $n-g$ of the variables due to $d\le g$.
However, this term is indeed necessary in \eqref{diag1} to ensure that the main formula \eqref{cH} 
remains neatly valid, including the diagonal terms.
We have thus completed the proof of Theorem \ref{th:main}.
 
\vspace{0.2cm}
It is natural to consider a generalization of the alternating sum in \eqref{cH} by introducing a parameter $\zeta$:
\begin{align}
\sum_{k=0}^{n+1}(-\zeta)^{k-1}\dot{T}^k(0).
\end{align}
Using \eqref{kakyo} and \eqref{kakyo2},  we find that its off-diagonal elements take a factorized form:
\begin{align}
\sum_{k=0}^{n+1}(-\zeta)^{k-1}
x_j \langle \boldsymbol{\sigma}'|\dot{T}^k(0)| \boldsymbol{\sigma}\rangle_j
=
\zeta^{d-1}(1-t)(1-\zeta t^{K_{\bar{h}_1}})
\cdots (1-\zeta t^{K_{\bar{h}_{n-g}}})\prod_{i=1}^g(1-t)t^{\ell_{h_i}} ,
\end{align} 
where notation follows \eqref{esum}.
In particular for $\zeta = t^{-K_1}, \ldots, t^{-K_n}$, this result reveals an interesting selection rule 
for nonzero transition coefficients in the process $\boldsymbol{\sigma}\rightarrow  \boldsymbol{\sigma}'$.
However, in general, these coefficients do not satisfy the positivity condition for off-diagonal transition rates.

\section{Further properties of $t$-PushTASEP}\label{sec:FP}

\subsection{Stationary eigenvalue of $T^k(z)$}\label{ss:Lak}

Let $|\overline{\mathbb{P}}({\bf m})\rangle  \in \mathbb{V}({\bf m})$ 
be the {\em stationary state} of the $t$-PushTASEP. 
It is a unique vector, up to normalization,  satisfying 
$H_{\text{PushTASEP}} |\overline{\mathbb{P}}({\bf m})\rangle=0$.
From Theorem \ref{th:main} and the commutativity \eqref{tcom}, it follows that 
$|\overline{\mathbb{P}}({\bf m})\rangle$ is a {\em joint} eigenvector of 
the transfer matrices $T^0(z),\ldots, T^{n+1}(z)$.
Moreover, while $|\overline{\mathbb{P}}({\bf m})\rangle$ depends on the inhomogeneities 
$x_1,\ldots, x_L$, it remains independent of $z$.
Let $\Lambda^k(z) = \Lambda^k(z|x_1,\ldots, x_L)$ be the stationary eigenvalue of $T^k(z)$, so that
$T^k(z)|\overline{\mathbb{P}}({\bf m})\rangle = 
 \Lambda^k(z)|\overline{\mathbb{P}}({\bf m})\rangle$.
Following an analytic Bethe ansatz argument similar to that in \cite[sec. 4.1]{KMMO16}, we obtain 
the following expression:\footnote{We omit a rigorous derivation in this paper. 
The result corresponds to the case where all Baxter $Q$ functions become constant, 
as demonstrated in \cite[Sec. 4.5]{KMMO16}. }
\begin{align}\label{Lam}
\Lambda^k(z|x_1,\ldots, x_L) = e_{k-1}(t^{K_1},\ldots, t^{K_n})\prod_{j=1}^L\left(1-\frac{tz}{x_j}\right)
+  e_{k}(t^{K_1},\ldots, t^{K_n})\prod_{j=1}^L\Bigl(1-\frac{z}{x_j}\Bigr),
\end{align}
where $K_i$ is defined in \eqref{Ki} and depends on ${\bf m}$.
This is a Yang-Baxterization of the $k$'th elementary symmetric polynomial:
\begin{subequations}
\begin{align}
\Lambda^k(z) &= \prod_{j=1}^L d_k \Bigl( \frac{z}{x_j} \Bigr)^{-1}\sum_{0 \le i_1 < \dots < i_k \le n}
\fbox{$i_1$}_{\,z} \,\fbox{$i_2$}_{\,t^{-1}z} \cdots \fbox{$i_k$}_{\,t^{-k+1}z},
\\
\fbox{$i$}_{\, z} &= t^{K_i}\prod_{j=1}^L\left(1-t^{\delta_{i,0}}\frac{z}{x_j}\right),
\end{align}
\end{subequations}
where $d_k(z)$ is defined by  \eqref{dk}.
For $k=0$ and $k=n+1$, the formula \eqref{Lam}  simplifies to \eqref{T0z} and \eqref{Tn1}, respectively, 
as $e_{k-1}(t^{K_1},\ldots, t^{K_n})$ and $e_{k}(t^{K_1},\ldots, t^{K_n})$ vanish.

Now let $\dot{\Lambda}(z) = \frac{d\Lambda(z)}{dz}$.
Differentiating \eqref{Lam} at $z=0$, we obtain 
\begin{align}
\dot{\Lambda}^k(0) = 
-\left(t e_{k-1}(t^{K_1},\ldots, t^{K_n})+e_{k}(t^{K_1},\ldots, t^{K_n})\right)
\sum_{j=1}^L\frac{1}{x_j}
=-e_k(t,t^{K_1},\ldots, t^{K_n})\sum_{j=1}^L\frac{1}{x_j}.
\end{align}
This {explains the} origin of the quantity $D_{\bf m}$ {in} \eqref{Dm} in the derivative of eigenvalues as 
\begin{align}\label{Dm2}
\sum_{k=0}^{n+1}(-1)^{k-1}\dot{\Lambda}^k(0)= 
\left( \sum_{j=1}^L\frac{1}{x_j} \right) \sum_{k=0}^{n+1}(-1)^{k}
 e_k(t,t^{K_1},\ldots, t^{K_n})
 =\left( \sum_{j=1}^L\frac{1}{x_j} \right) D_{\bf m}.
\end{align}
Consequently, our main formula \eqref{cH} is also expressed as
\begin{align}\label{ccH}
H_\text{PushTASEP}(x_1,\ldots, x_L)  = 
D_{\bf m}^{-1}\frac{d}{dz}\sum_{k=0}^{n+1} (-1)^{k-1}
\left.
\left(T^k(z|x_1,\ldots, x_L) - \Lambda^k(z|x_1,\ldots, x_L)\right)\right|_{z=0}.
\end{align}
From this, the stationarity condition 
\begin{align}\label{HP0}
H_\text{PushTASEP}(x_1,\ldots, x_L) |\overline{\mathbb{P}}({\bf m})\rangle=0
\end{align}
becomes evident.

We now present some examples of (unnormalized) stationary states.

\begin{example}
In what follows, \emph{cyc.} means the terms that are generated by
cyclic permutations in $\Z_L$ taking
$|\sigma_1,\ldots, \sigma_L\rangle \rightarrow 
|\sigma_L,\ldots, \sigma_{L-1}\rangle$ with 
$x_i \rightarrow x_{i+1}$ ($i \mod L$).
For ${\bf m}=(1,1,1)$, the stationary state is given by 
\[
\frac{t x_1 + x_3 + t x_3}{x_1} |012\rangle
+ \frac{x_2 + x_3 + t x_3}{x_2} |102\rangle+ \text{cyc.},
\]
for ${\bf m}=(1,2,1)$, it is
\begin{multline*}
\frac{t^2 x_1 + x_4 + t x_4 + t^2 x_4}{x_1} |0112\rangle
+ \frac{t x_2 + x_4 + t x_4 + t^2 x_4}{x_2} |1012\rangle \\
+ \frac{x_3 + x_4 + t x_4 + t^2 x_4}{x_3} |1102\rangle + \text{cyc.},
\end{multline*}
and for ${\bf m}=(2,2,1)$, it is
\begin{multline*}
\frac{t^2 x_1 + t^2 x_2 + x_5 + t x_5 + t^2 x_5}{x_1 x_2} |00112\rangle
+ \frac{t^2 x_1 + t x_3 + x_5 + t x_5 + t^2 x_5}{x_1 x_3} |01012\rangle
\\
+ \frac{t x_2 + t x_3 + x_5 + t x_5 + t^2 x_5}{x_2 x_3} |10012\rangle
+ \frac{t^2 x_1 + x_4 + x_5 + t x_5 + t^2 x_5}{x_1 x_4} |01102\rangle
\\
+ \frac{t x_2 + x_4 + x_5 + t x_5 + t^2 x_5}{x_2 x_4} |10102\rangle
+ \frac{x_3 + x_4 + x_5 + t x_5 + t^2 x_5}{x_3 x_4} |11002\rangle+ \text{cyc.}
\end{multline*}
\end{example}

\subsection{Matrix product  formula for the stationary probability}\label{ss:mps}

As remarked in the previous subsection, our Theorem \ref{th:main} reduces 
the problem of finding the stationary probability of the inhomogeneous $n$-species $t$-PushTASEP
to that for a discrete time Markov process whose time evolution is governed by 
the (suitably normalized) transfer matrix $T^1(z|x_1,\ldots, x_L)$.
Here, we present a simple derivation of the matrix product formula for the stationary probability
based on $T^1(z|x_1,\ldots, x_L)$.

Matrix product formulas were first obtained for homogeneous $n$-species ASEP in \cite{PEM09} 
using operators defined by nested recursion relations.
An inhomogeneous extension was introduced in \cite{CDW15}
in connection with the Zamolodchikov-Faddeev algebra and Macdonald polynomials.
Further developments on matrix product operators were explored in \cite{KOS24},  
where the nested recursive structure is identified with the multiline queue construction \cite{CMW22}
culminating in a corner transfer matrix formulation of a quantized five-vertex model.
It allows for the simplest diagrammatic representation devised to date, 
with a natural three-dimensional interpretation.\footnote{The graphical representation in \cite{CDW15} 
needs an $n$-color pen, whereas the five-vertex model formulation in \cite{KOS24} uses only two states $0$ and $1$.} 
Our presentation here is based on \cite{KOS24}.

Let $X_0(z), \ldots, X_n(z)$ be the ``corner transfer matrices" defined in \cite[Def.15]{KOS24}.\footnote{This is an abuse 
of terminology from \cite[Chap.13]{Bax82}, where it is defined for a two-dimensional lattice. Unlike in that context, 
$X_i(z)$  here acts in the direction of a third dimension. }
These are linear operators depending on the spectral parameter $z$, 
and act on the $\frac{n(n-1)}{2}$-fold tensor product of $t$-oscillator Fock spaces.
To align with the convention used for $R(z)^{\alpha, \beta}_{\gamma,\delta}$ in \cite[eq. (16)]{KOS24}
and $S^{a,b}_{i,j}(z)$ in \eqref{Sb2}, we adopt the index transformation  
$0,1,\ldots,n \rightarrow n,\ldots, 1, 0$.\footnote{In this section, we use the simplified notation $S(z)^{a,b}_{i,j}$ for 
$S(z)^{{\bf e}_a, {\bf e}_b}_{{\bf e}_i, {\bf e}_i}$ as introduced in Section \ref{ss:fusion}.}
Further inverting $z$, we set $A_\alpha(z) = X_{n-\alpha}(z^{-1})$ for $0 \le \alpha \le n$.
The key result required here is \cite[Th.28]{KOS24}, which states that the following 
Zamolodchikov-Faddeev algebra holds:
 \begin{align}\label{zf}
 \Bigl( 1- \frac {tz}x \Bigr) A_\alpha(x) A_\beta(z) = \sum_{\gamma, \delta=0}^n
 S \Bigl( \frac zx \Bigr)^{\beta, \alpha}_{\gamma,\delta}A_\gamma(z)A_\delta(x).
 \end{align}
Let us introduce a vector whose coefficients are given in the matrix product (mp) form:
\begin{subequations}
\begin{align}
&|\mathbb{P}_\mathrm{mp}\rangle = \sum_{(\sigma_1,\ldots, \sigma_L) \in \mathcal{S}({\bf m})}
{\mathbb{P}}_\mathrm{mp}(\sigma_1,\ldots, \sigma_L)
|\sigma_1,\ldots, \sigma_L\rangle \in \mathbb{V}({\bf m}),
\label{pvec}
\\
&{\mathbb{P}}_\mathrm{mp}(\sigma_1,\ldots, \sigma_L)
=\mathrm{Tr}\left(A_{\sigma_1}(x_1)\cdots A_{\sigma_L}(x_L)\right),
\label{pA}
\end{align}
\end{subequations}
where $\mathbb{V}({\bf m})$ and $\mathcal{S}({\bf m})$ are defined in 
\eqref{Vm} and \eqref{Sm}, respectively.
The  trace is nonzero and convergent under the assumption $m_0,\ldots, m_n \ge 1$.

\begin{proposition}\label{pr:T1}
The vector $|\mathbb{P}_\mathrm{mp}\rangle$ is an eigenvector of 
$T^1(z)$ with eigenvalue $\Lambda^1(z)$ given by \eqref{Lam}.
That is,
\begin{align}\label{tpL}
T^1(z|x_1,\ldots, x_L)|\mathbb{P}_\mathrm{mp}\rangle
= \Lambda^1(z|x_1,\ldots, x_L)|\mathbb{P}_\mathrm{mp}\rangle.
\end{align}
\end{proposition}
\begin{proof}
From \eqref{Srepeat0}, \eqref{tke} and \eqref{Lam}, 
the difference between the two sides of \eqref{tpL} is a  polynomial in $z$ of degree at most $L$. 
Therefore it suffices to check the equality at the $L+1$ points
$z=0,x_1,\ldots, x_L$.
At $z=0$, it follows from \eqref{tk0}, \eqref{Lam}  
and $e_{k-1}(t^{K_1},\ldots, t^{K_n}) + e_k(t^{K_1},\ldots, t^{K_n})
= e_k(t^{K_0},t^{K_1},\ldots, t^{K_n})$ with $k=1$. (Note $K_0=0$.) 
To verify the equality at the other points,  we employ a standard approach.
We begin by computing the action of $T^1(z|x_1,\ldots, x_L)$ using \eqref{tke}:
\begin{align}
T^1(z|x_1,\ldots, x_L)|\mathbb{P}_\mathrm{mp}\rangle 
=&\sum_{(\sigma_1,\ldots, \sigma_L) \in \mathcal{S}({\bf m})}
{\mathbb{P}}'_\mathrm{mp}(\sigma_1,\ldots, \sigma_L)|\sigma_1,\ldots, \sigma_L\rangle,
\\
{\mathbb{P}}'_\mathrm{mp}(\sigma_1,\ldots, \sigma_L) 
=& \sum_{\substack{a_1,\ldots, a_L \in \{0,\ldots, n\}  \\ (\sigma'_1,\ldots, \sigma'_L) \in \mathcal{S}({\bf m})} }
\!\!\!S \Bigl( \frac{z}{x_1} \Bigr)^{a_2,\sigma_1}_{a_1, \sigma'_1}
S \Bigl( \frac{z}{x_2} \Bigr)^{a_3,\sigma_2}_{a_2, \sigma'_2}
\cdots 
S \Bigl( \frac{z}{x_L} \Bigr)^{a_1,\sigma_L}_{a_L, \sigma'_L} \cr
& \times \mathrm{Tr}(A_{\sigma'_1}(x_1) \cdots A_{\sigma'_L}(x_L)).
\label{ppmp}
\end{align}
The summations over $(\sigma_1,\ldots, \sigma_L)$ 
and $(\sigma'_1,\ldots, \sigma'_L)$ are restricted to $\mathcal{S}({\bf m})$
by the weight conservation property of $T^1(z)$ and $S(z)$. 
 Now, consider the specialization $z=x_1$.
 From \eqref{z1}, the leftmost factor $S(z/x_1)^{a_2,\sigma_1}_{a_1, \sigma'_1}$ in \eqref{ppmp} simplifies to
 $(1-t)\delta_{a_1,\sigma_1}\delta_{a_2, \sigma'_1}$.
Substituting this into the RHS of \eqref{ppmp} gives
 \begin{align}
(1-t)
 \sum_{\substack{a_3,\ldots, a_L   \\ \sigma'_1,\ldots, \sigma'_L} }
\!\!\!S \Bigl( \frac{x_1}{x_2} \Bigr)^{a_3,\sigma_2}_{\sigma'_1, \sigma'_2}
S \Bigl( \frac{x_1}{x_3} \Bigr)^{a_4,\sigma_3}_{a_3, \sigma'_3}
\cdots 
S \Bigl( \frac{x_1}{x_L} \Bigr)^{\sigma_1,\sigma_L}_{a_L, \sigma'_L}
\mathrm{Tr}(A_{\sigma'_1}(x_1)A_{\sigma'_2}(x_2) \cdots A_{\sigma'_L}(x_L)).
\end{align}
Applying \eqref{zf}, we sum over $\sigma'_1, \sigma'_2$ obtaining
\begin{multline}
(1-t) {\left( 1-\frac{tx_1}{x_2} \right)}
 \sum_{\substack{a_3,\ldots, a_L   \\ \sigma'_1,\ldots, \sigma'_L} }
S \Bigl( \frac{x_1}{x_3} \Bigr)^{a_4,\sigma_3}_{a_3, \sigma'_3}
\cdots 
S \Bigl( \frac{x_1}{x_L} \Bigr)^{\sigma_1,\sigma_L}_{a_L, \sigma'_L} \\
 \times \mathrm{Tr}(A_{\sigma_2}(x_2)A_{a_3}(x_1)A_{\sigma'_3}(x_3) \cdots A_{\sigma'_L}(x_L)).
\end{multline}
We can successively push $A_\bullet(x_1)$ with any index $\bullet$ to the right using \eqref{zf},
leading to
\begin{multline}
\prod_{j=1}^L\left(1-\frac{tx_1}{x_j}\right)
\mathrm{Tr}\left(
A_{\sigma_2}(x_2)A_{\sigma_3}(x_3)\cdots A_{\sigma_L}(x_L)A_{\sigma_1}(x_1)\right)
\\
=\Lambda^1(x_1|x_1,\ldots, x_L) 
{\mathbb{P}}_\mathrm{mp}(\sigma_1,\ldots, \sigma_L).
\end{multline}
For $z = x_i$ in general, the proof follows analogously due to the cyclicity of the trace. 
Namely,  $A_\bullet(x_i)$ becomes ``active" and circulates within the trace successively replacing  each 
$S(x_i/x_j)$  by $(1-tx_i/x_j)$ until it returns to its original position.
\end{proof}

The dynamics of particles circulating in a one-dimensional system via $R$-matrices, 
as observed in the final step of the proof,  
dates back to \cite[eq. (14)]{Y67} and is sometimes referred to as  \emph{Yang's system}.

From Proposition \ref{pr:T1}, 
it follows that the matrix product state \eqref{pvec} is a joint eigenstate 
of $T^0(z), \ldots, \allowbreak T^{n+1}(z)$.
Given their eigenvalues $\Lambda^0(z), \ldots, \Lambda^{n+1}(z)$ as in \eqref{Lam}, along with 
the result \eqref{HP0} and the uniqueness of the stationary state,  we conclude that 
\eqref{pA} provides a matrix product formula for the (unnormalized) stationary probability 
of the inhomogeneous $n$-species $t$-PushTASEP.
 
 \section{ASEP Markov matrix from transfer matrix}\label{sec:asep}
 For readers convenience, we include a short elementary section recalling the well-known 
origin of the $n$-species ASEP Markov matrix in a commuting family of transfer matrices
 in the convention of this paper.
 The ASEP is another Markov process on each sector $\mathbb{V}({\bf m})$ in \eqref{Vm}.
 Its Markov matrix consists of the nearest neighbor interaction terms as
 \begin{subequations}
 \begin{align}
 H_{\text{ASEP}} &= \sum_{i \in \Z_L} 1 \otimes \cdots \otimes  1 \otimes 
 h_{\text{ASEP}} \otimes 1 \otimes \cdots \otimes 1,
\label{hasep1}\\
 h_{\text{ASEP}}(\vv_\alpha \otimes \vv_\beta) 
 &=(\vv_{\beta} \otimes \vv_\alpha  - \vv_\alpha \otimes \vv_\beta )t^{[\alpha>\beta]}.
 \label{hasep2}
 \end{align}
 \end{subequations}
 where in \eqref{hasep1}, $h_{\text{ASEP}}$ acts on the $(i,i+1)$ components of $\VV^{\otimes L}$.
 It swaps the local states $0 \le \alpha, \beta \le n$ in adjacent sites with the rate $t^{[\alpha>\beta]}$.
 
 Let us consider  the transfer matrix $T^k(z)$ in  \eqref{tkdef}--\eqref{tke} 
 in the special case $k=1$ with the homogeneous choice of parameters 
 $x_1=\cdots = x_L=1$.
We denote it as $T^1(z|{\bf x}={\bf 1})$.
As a corollary of \eqref{tcom}, they still satisfy the commutativity:
\begin{align}
[T^1(z|{\bf x}={\bf 1}),T^1(z'|{\bf x}={\bf 1})]=0.
\end{align}
 
For the simplest $R$-matrix $S(z)= S^{1,1}(z)$ in \eqref{Sb2}, which is relevant to 
$T^1(z|{\bf x}={\bf 1})$, it is straightforward to check  
\begin{align}
S(1) &= (1-t)\mathcal{P}, \quad \mathcal{P}(\vv \otimes \vv') = \vv' \otimes \vv,
\label{z1}
\\
\mathcal{P}\!\left. \frac{dS(z)}{dz}\right|_{z=1} &= - h_{\text{ASEP}} -t\,\mathrm{Id}.
\label{dsz}
\end{align}
From \eqref{z1}, one finds that $T^1(1|{\bf x}={\bf 1}) = (1-t)^L \mathcal{C}$, where 
 $\mathcal{C}(\vv_{\sigma_1} \otimes \vv_{\sigma_2} \otimes \cdots \otimes \vv_{\sigma_L}) = 
 \vv_{\sigma_L} \otimes \vv_{\sigma_1} \otimes \cdots \otimes \vv_{\sigma_{L-1}}$ represents  a cyclic shift.
 Using this result,  \eqref{dsz} leads, via an argument analogous to \cite[eq. (55)]{KMMO16},  to 
 \begin{align}\label{ht1}
 H_{\text{ASEP}} = -(1-t)\frac{d}{dz} \left. \log T^1(z|{\bf x}={\bf 1}) \right|_{z=1}-tL\, \mathrm{Id},
 \end{align}
 which is an example of the classic Baxter's formula for deducing Hamiltonians from commuting 
 transfer matrices \cite[eq. (10.14.20)]{Bax82}.
 
Recall the joint eigenvector $|\overline{\mathbb{P}}({\bf m})\rangle$ of
$T^0(z), \ldots, T^{n+1}(z)$ with eigenvalues $\Lambda^0(z), \ldots,
\allowbreak \Lambda^{n+1}(z)$
introduced in Section \ref{ss:Lak}.
They all depend on the inhomogeneities $x_1, \ldots, x_L$.
From the specialization $x_1=\cdots = x_L=1$  and \eqref{Lam}, we have
\begin{subequations}
\begin{align}
T^1(z|{\bf x}={\bf 1})|\overline{\mathbb{P}}({\bf m})\rangle_{{\bf x}={\bf 1}}
&=\Lambda^1(z|{\bf x}={\bf 1})|\overline{\mathbb{P}}({\bf m})\rangle_{{\bf x}={\bf 1}},
\label{tpm}\\
\Lambda^1(z|{\bf x}={\bf 1}) &= (1-tz)^L + e_1(t^{K_1},\ldots, t^{K_n})(1-z)^L.
\end{align}
\end{subequations}
Using 
$(1-t)\frac{d}{dz} \left. \log \Lambda^1(z|{\bf x}={\bf 1}) \right|_{z=1}= -tL$ along with
\eqref{tpm} and \eqref{ht1}, one can check the stationarity condition
$H_{\text{ASEP}} |\overline{\mathbb{P}}({\bf m})\rangle_{{\bf x}={\bf 1}}=0$
as desired.
The (unnormalized) stationary probability is given by the matrix product formula \eqref{pA} 
with the homogeneous specialization $x_1= \cdots = x_L=1$.

\appendix

\section{Proof of Theorem \ref{th:S}\label{app:proof}}

Here we prove
\begin{align}
S(z)^{{\bf a}, {\bf e}_b}_{{\bf i}, {\bf e}_j} = \delta^{{\bf a}+{\bf e}_b}_{{\bf i} + {\bf e}_j}
 (-1)^{\# \{s \in [1, k] \mid A_s < j\} + \# \{s \in [1, k] \mid I_s < b\}} \;
 t^{\# \{s \in [1, k] \mid A_s > j\}} \;
 (1-t^{a_j}z^{\delta_{b,j}})z^{[j>b]},
 \label{Srepeat}
\end{align}
where we have expressed  \eqref{Srepeat0} in terms of $A, I  \in \mathscr{T}^k$  which 
correspond to the tableau representations of ${\bf a}, {\bf i} \in \BB^k$.

\begin{proof}
We will perform induction on $k$. For $k = 1$, there is only a single term in the sum. 
According to \eqref{Srepeat}, the answer should be
\[
(-1)^0 t^{[A_1 > j]} z^{[j > b]} (1 - t^{[A_1 = j]} z^{[b = j]}).
\]
There are three cases:
\begin{enumerate}
\item $I_1 = A_1 = b = j$. We then get $1 - t z$.

\item $I_1 = A_1 = \alpha$ (say) and $b = j = \beta$ (say), with $\alpha \neq \beta$. In this case, we get 
$t^{[\alpha > \beta]} (1 - z)$.

\item $I_1 = b = \alpha$ (say) and $A_1 = j = \beta$ (say), with $\alpha \neq \beta$. In this case, we get 
$z^{[\beta > \alpha]} (1 - t)$.
\end{enumerate}
\noindent
All these weights match with Figure \ref{fig:sb}, completing the proof in this case.

Now suppose the results holds for $k-1$. That is to see that for all fixed tuples $I' = (I_2, \dots, I_k), A' = (A_2, \dots, A_k) \in \mathscr{T}^{k-1}$ and elements $j', b \in \mathscr{B}^1$, 
the weight $S(z)^{{\bf A'}, {\bf e}_b}_{{\bf I'}, {\bf e}_{j'}}$ given in Figure \ref{fig:matrix_element} 
is equal to \eqref{Srepeat} with $k$ replaced by $k-1$. 
We now consider the $k$-weight, which is like adding one more row at the bottom to the diagram in 
Figure \ref{fig:matrix_element}. 
For consistency with the induction hypothesis, let $j'$ be the label attached to the vertical line between the bottom two rows.

There are two cases to consider. First, suppose $A_1 \neq j$. In that case, we must have $j' = j$ and $\sigma(I_1) = A_1$, 
which is the smallest among the $A_i$'s. 
The weight of the vertex in the bottom row is $t^{[A_1 > j]} (1 - t^{-k+1} z)$.
Therefore, the sum is over all permutations in $\mathfrak{S}_{k-1}$. We can now apply the induction hypothesis and the weight of the remainder of the diagram is
\begin{multline*}
\frac{1}{1 - t^{-k+1} z} (-1)^{\# \{s \in [2, k] \mid A_s < j\} + \# \{s \in [2, k] \mid \sigma(I_s) < b\} 
+ \# \{s \in [1, k] \mid I_s < A_1\}} \\
\times t^{\# \{s \in [2, k] \mid A_s > j\}} 
\; z^{[j > b]}
\left( 1 - t^{[j \in \{A_2, \dots, A_k\}]} z^{[b = j]} \right),
\end{multline*}
where the last term in the sign arises from the sign of permutations in which $\sigma(I_1) = A_1$.
Combining these factors, the power of $t$ is
\[
[A_1 > j] + \# \{s \in [2, k] \mid A_s > j\}
 = \# \{s \in [1, k] \mid A_s > j\},
\]
which matches the power in \eqref{Srepeat}.
The factor of $1 - t^{-k+1} z$ cancels and the power of $z$ is unchanged. 

The last thing we need to check is the sign. To that end, note that
\[
\# \{s \in [1, k] \mid A_s < j\} = \# \{s \in [2, k] \mid A_s < j\} + [A_1 < j],
\]
and
\[
\# \{s \in [1, k] \mid \sigma(I_s) < b\} = \# \{s \in [2, k] \mid \sigma(I_s) < b\}
+ [\sigma(I_1) < b].
\]
If $j < A_1$, then $j < A_1, \dots, A_k$ and so $j = b$ since it must exit somewhere. Thus $[A_1 < j] = [\sigma(I_1) < b] = 0$.
If $A_1 < j$ and $j = b$, then $[A_1 < j] = [\sigma(I_1) < b] = 1$ trivially.
In both these cases, $\# \{s \in [1, k] \mid I_s < A_1\} = 0$ since $\sigma(I_1) = A_1$ is the smallest entry in $I$.
If $A_1 < j$ and $j \neq b$, then $j = A_u$ for some $u \in [2, k]$.
Now, $\# \{s \in [1, k] \mid I_s < A_1\} = [A_1 > b]$. Therefore, the total extra sign is the parity of $[A_1 < j] + [A_1 < b] + [A_1 > b]$, which is even.
Thus, in all cases the signs match when $A_1 \neq j$.

We now come to the nontrivial case, namely when $A_1 = j$. Let $j'$ be the label on the vertical edge immediately above the lowest vertex. By conservation 
$j' = \sigma(I_1)$. There are now two kinds of contributions to the sum in 
Figure \ref{fig:matrix_element}. First, suppose the permutation $\sigma$ is such that 
$\sigma(I_1) = j' = j$. Then all four edges incident to the lowest vertex have label $j$. This has weight $1 - t^{-k+2} z$ by 
Figure \ref{fig:sb}.
By the induction hypothesis, summing over all permutations in $\mathfrak{S}_{k-1}$
for the other vertices, and multiplying by this weight gives
\begin{multline}
\label{caseiia}
\frac{1 - t^{-k+2} z}{1 - t^{-k+1} z} (-1)^{\# \{s \in [2, k] \mid A_s < j\} + \# \{s \in [2, k] \mid \sigma(I_1) = j, \sigma(I_s) < b\}}
\\ 
\times \;t^{\# \{s \in [2, k] \mid A_s > j\}} 
\; z^{[j > b]}
\left( 1 - t^{[j \in \{A_2, \dots, A_k\}]} z^{[b = j]} \right).
\end{multline}

Now consider permutations $\sigma$ such that $\sigma(I_1) = j' \neq j$. 
By Figure \ref{fig:sb}, the lowest vertex has weight $(1-t) \, (z t^{-k+1})^{[j > j']}$.
Apply the induction hypothesis to the configuration where the lowest vertex has label $j'$, and multiply by this weight to get the factor
\begin{multline}
\label{caseiib}
\frac{(1-t) \, (z t^{-k+1})^{[j > j']}}{1 - t^{-k+1} z} 
(-1)^{\# \{s \in [2, k] \mid A_s < j'\} + \# \{s \in [2, k] \mid \sigma(I_1) = j',  \sigma(I_s) < b\} + \# \{s \in [1, k] \mid I_s < j'\}} \\
\times \;t^{\# \{s \in [2, k] \mid A_s > j'\}} 
\; z^{[j' > b]}
\left( 1 - t^{[j' \in \{A_2, \dots, A_k\}]} z^{[b = j']} \right).
\end{multline}
Notice that there is an extra contribution to the sign in \eqref{caseiib} because of the sign of the permutation $\sigma$ in Figure \ref{fig:matrix_element}.
We want to analyze the sum of \eqref{caseiia} and \eqref{caseiib}. This has to be done on a case-by-case basis.

First, suppose $j = b$. Then \eqref{caseiia} becomes
\begin{equation}
\label{iia1}
\frac{1 - t^{-k+2} z}{1 - t^{-k+1} z} (-1)^{0 + 0} t^{k-1} z^0 (1 - t^0 z^1)
= \frac{t^k (t^{k-2} - z) (1 - z)}{t^{k-1} - z}.
\end{equation}
We have to sum \eqref{caseiib} over all possible values of $u \in [2, k]$ such that $j' = A_u$ since $j' \neq b$. Thus, $j' > j$ and we obtain
\begin{equation}
\label{iib1}
\sum_{u = 2}^k \frac{(1-t)}{1 - t^{-k+1} z} (-1)^{(u-2) + 0 + (u-1)} t^{k-u}
z^1 (1 - t^1 z^0)
=
- \frac{t^{k-1} (1-t)^2 z}{t^{k-1} - z} \sum_{u = 2}^k  t^{k-u},
\end{equation}
which sums to 
\begin{equation}
\label{iib2}
- \frac{t^{k-1} (1-t) (1-t^{k-1}) z}{t^{k-1} - z}.
\end{equation}
Summing \eqref{iia1} and \eqref{iib2} gives $t^{k-1}(1 - tz)$, which matches \eqref{Srepeat} for $A_1 = j = b$.

Finally, suppose $j \neq b$. Then the calculation depends on the relative order of $j$ and $b$. 

If $b < j$, then $j'$ cannot equal $j$ because $j \notin \{A_2, \dots, A_k, b\}$. Therefore, \eqref{caseiia} cannot contribute. We sum \eqref{caseiib} over all possible values of $u \in [2, k]$ such that $j' = A_u$, and in addition also consider $j' = b$.
In the former case, we obtain exactly \eqref{iib1}, and in the latter,
\[
\frac{(1-t) \, (z t^{-k+1})}{1 - t^{-k+1} z} (-1)^{0 + 0 + 0} t^{k-1}
z^0 (1 - t^0 z^1) = \frac{t^{k-1} (1-t) (1-z) z}{t^{k-1} - z}.
\]
Summing this with \eqref{iib1} gives, after some simplifications, $t^{k-1} (1-t)z$, which matches \eqref{Srepeat} for $A_1 = j > b$.

The last case is when $j < b$. Notice that $A_u$ cannot be equal to $b$ for any $u \in [2, k]$ by conservation. Therefore $b$ is a label different from $\{A_1, \dots, A_k\}$. Thus, again by conservation, $j'$ cannot equal $j$ and so \eqref{caseiia} does not contribute. So, we need to look at \eqref{caseiib}. Suppose $b$ is such that $A_v < b < A_{v+1}$ for some $v \in [k]$, where we interpret $v = k$ as saying that $b > A_k$.
As in the situation immediately above, $j'$ can equal $A_u$ for some $u \in [2, k]$ or $j'$ can equal $b$. In the former case, we get
\begin{multline*}
\sum_{u = 2}^k \frac{(1-t)}{1 - t^{-k+1} z} (-1)^{(u-2) + (v - 2 + [u > v]) + (u-2 + [u > v])} t^{k-u}
z^{[u > v]} (1 - t^1 z^0) \\
=
(-1)^v \frac{t^{k-1} (1-t)^2 }{t^{k-1} - z} 
\left( -\sum_{u = 2}^v  t^{k-u} - z\sum_{u = v+1}^k  t^{k-u} \right) ,
\end{multline*}
which sums to 
\begin{equation}
\label{iib3}
(-1)^v \frac{t^{k-1} (1-t) }{t^{k-1} - z} 
\left( t^{k-v}(1 - t^{v-1}) + z (1 - t^{k-v}) \right).
\end{equation}
When $j' = b$, we get
\[
\frac{(1-t)}{1 - t^{-k+1} z} (-1)^{(v-1) + (v-1) + v-1} t^{k-v}
z^0 (1 - t^0 z^1) = (-1)^{v-1} \frac{t^{2k-1-v} (1-t) (1-z)}{t^{k-1} - z}.
\]
Summing this with \eqref{iib3} gives $(-1)^{v-1} t^{k-1} (1-t)$, which again matches \eqref{Srepeat} for $A_1 = j < b$.

We have thus verified all the cases for the boundary labels, completing the proof.
\end{proof}

\section{3D construction of $S^{k,1}(z)$\label{app:3dR}}

We introduce the operators $\LL=(\LL^{a,b}_{i,j})_{a,b,i,j \in \{0,1\}}$ by
\begin{subequations}
\begin{align}
&\LL^{0,0}_{0,0} = \LL^{1,1}_{1,1} = 1, \quad
\LL^{1,0}_{1,0} = {\rm {\bf k}}, \quad
\LL^{0,1}_{0,1} = -q {\rm {\bf k}}, \quad
\LL^{1,0}_{0,1} = {\rm {\bf a}}^+,\quad
\LL^{0,1}_{1,0} = {\rm {\bf a}}^-,
\label{Lop}
\\
&\LL^{a,b}_{i,j}=0 \quad  \text{if}\;\;  a+b \neq i+j.
\label{Lc}
\end{align}
\end{subequations}
Here $\ok, \ap, \am$ are $q$-oscillator operators\footnote{The parameter $q$
will be set to $t^{1/2}$ in \eqref{S02}.}  on the Fock space 
$F= \bigoplus_{m \ge 0}\Q(q) |m\rangle$, defined by 
\begin{align}
\ok | m \rangle = q^m |m \rangle,\quad 
\ap | m \rangle = |m+1\rangle,\quad 
\am | m \rangle = (1-q^{2m})|m-1\rangle.
\end{align}
We will also use  the ``number operator"  ${\bf h}$  on $F$ acting as  
${\bf h}|m\rangle = m|m\rangle$. 
Thus $\ok = q^{\bf h}$.
One may regard $\LL$ as  defining 
a $q$-oscillator-weighted six-vertex model as in  Figure \ref{fig:6v}.

\begin{figure}[h]
\begin{center}
\begin{picture}(320,60)(-20,30)
{\unitlength 0.011in
\put(12,80){
\put(-11,0){\vector(1,0){23}}\put(0,-10){\vector(0,1){22}}
}
\multiput(80,80)(55,0){6}{
\put(-11,0){\vector(1,0){23}}\put(0,-10){\vector(0,1){22}}
}
\put(-68,0){
\put(60.5,77){$i$}\put(77.5,60){$j$}
\put(96,77){$a$}\put(77.5,96.5){$b$}
}
\put(61,77){0}\put(77.5,60){0}\put(94,77){0}\put(77.5,96.5){0}
\put(55,0){
\put(61,77){1}\put(77.5,60){1}\put(94,77){1}\put(77.5,96.5){1}
}
\put(110,0){
\put(61,77){1}\put(77.5,60){0}\put(94,77){1}\put(77.5,96.5){0}
}
\put(165,0){
\put(61,77){0}\put(77.5,60){1}\put(94,77){0}\put(77.5,96.5){1}
}
\put(220,0){
\put(61,77){0}\put(77.5,60){1}\put(94,77){1}\put(77.5,96.5){0}
}
\put(275,0){
\put(61,77){1}\put(77.5,60){0}\put(94,77){0}\put(77.5,96.5){1}
}
\put(78,40){
\put(-74,0){$\LL^{a,b}_{i,j}$}
\put(0,0){1} \put(55,0){1} \put(109,0){${\rm{\bf k}}$}
\put(157,0){$-q{\rm{\bf k}}$} \put(218,0){${\rm{\bf a}}^+$} \put(274,0){${\rm{\bf a}}^-$}
}}
\end{picture}
\caption{$\LL=(\LL^{a,b}_{i,j})$ as a $q$-oscillator-weighted six-vertex model.
The $q$-oscillators may be regarded as acting along the third arrow perpendicular to the sheet in each vertex.
In this context, $\LL$ is referred to as a 3D $L$-operator.
This figure, which will only be used in this Appendix, should not be confused with Figure \ref{fig:sb}.}
\label{fig:6v}
\end{center}
\end{figure}

For $0 \le k,l\le n+1$, we introduce the linear map 
$R(z) = R^{k,l}(z) \in \mathrm{End}(V^k \otimes V^l)$ by
\begin{subequations}
\begin{align}
R(z)(v_{\bf i} \otimes v_{\bf j}) &=\sum_{{\bf a}\in \BB^k, {\bf b} \in \BB^l}
R(z)^{{\bf a}, {\bf b}}_{{\bf i}, {\bf j}}\,v_{\bf a} \otimes v_{\bf b}
\qquad ({\bf i} \in \BB^k, {\bf j} \in \BB^l),
\\
R(z)^{{\bf a}, {\bf b}}_{{\bf i}, {\bf j}} & = 
\mathrm{Tr}_F(z^{\bf h}\LL^{a_n, b_n}_{i_n,  j_n} \cdots \LL^{a_0, b_0}_{i_0,  j_0}).
\label{trL}
\end{align}
\end{subequations}
The trace is convergent as a formal power series in $q$ and $z$. 
From  \eqref{Lc}, $R(z)$ has the weight conservation property:
\begin{align}\label{rwc}
R(z)^{{\bf a}, {\bf b}}_{{\bf i}, {\bf j}} = 0 \;\,\text{unless}\;\,
{\bf a} + {\bf b} = {\bf i} + {\bf j}.
\end{align}
The cases $k=0,n+1$ reduce to the scalar matrices as
\begin{align}
R^{0,l}(z)&= \mathrm{Tr}_F\left(z^{\bf h}(\LL^{0,0}_{0,0})^{n+1-l}(\LL^{0,1}_{0,1})^l\right) \mathrm{Id}
=\frac{(-q)^l}{1-q^lz}\mathrm{Id},
\\
R^{n+1,l}(z) &= \mathrm{Tr}_F\left(z^{\bf h}(\LL^{1,1}_{1,1})^{l}(\LL^{1,0}_{1,0})^{n+1-l}\right)\mathrm{Id}
= \frac{1}{1-q^{n+1-l}z}\mathrm{Id}.
\end{align}
The first nontrivial case is $l=1$.
We express the elements ${\bf b}, {\bf j} \in \BB^1$ as ${\bf e}_b, {\bf e}_j$
with $0 \le b,j \le n$.
Then \eqref{trL} is evaluated explicitly as 
\begin{align}\label{re1}
(1-q^{k-1}z)(1-q^{k+1}z)R(z)^{{\bf a}, {\bf e}_b}_{{\bf i}, {\bf e}_j} & = 
\delta^{{\bf a}+{\bf e}_b}_{{\bf i} + {\bf e}_j} \times
\begin{cases}
(-q)^{1-a_j}(1-q^{2a_j+k-1}z), & j=b, 
\\
q^{a_{j+1}\cdots + a_{b-1}}(1-q^2),  &j<b,
\\
q^{k-1-(a_{b+1}+\cdots + a_{j-1})}(1-q^2)z, & j>b.
\end{cases}
\end{align}

It is known \cite{BS06,K22} that $\LL$ satisfies the tetrahedron equation,
a three dimensional generalization of the Yang-Baxter equation, of the form 
$\mathscr{R}_{456}\LL_{124}\LL_{135}\LL_{236} 
= \LL_{236}\LL_{135}\LL_{124}\mathscr{R}_{456}$ for some three dimensional $R$ matrix 
$\mathscr{R}$. 
By a projection onto the two dimension, it generates a family of Yang-Baxter equations:
\begin{align}\label{ybe1}
R^{k_1,k_2}_{1,2}(x)R^{k_1,k_3}_{1,3}(xy)R^{k_2,k_3}_{2,3}(y) = 
R^{k_2,k_3}_{2,3}(y) R^{k_1,k_3}_{1,3}(xy)R^{k_1,k_2}_{1,2}(x),
\end{align}
where $k_1,k_2,k_3 \in \{0,\ldots, n+1\}$.
They are equalities in $\mathrm{End}(V^{k_1}\otimes V^{k_2} \otimes V^{k_3})$
on which $R^{k_i,k_j}_{i,j}(z)$ acts on the $i$'th and the $j$'th components
as $R^{k_i,k_j}(z)$ and identity elsewhere. 
Details can be found in \cite[Chap. 11]{K22}.

\subsection{Modifying $R(z)$ into $S(z)$}

Let us proceed to a special gauge of the $R$-matrix relevant to the $t$-PushTASEP.
Following \cite[eq. (15)]{KMMO16}, we first introduce 
${\tilde S}(z) = {\tilde S}^{k,l}(z) \in \mathrm{End}(V^k \otimes V^l)$ by
\begin{align}
{\tilde S}(z)(v_{\bf i} \otimes v_{{\bf j}}) &=\sum_{{\bf a}\in \BB^k,  {\bf b} \in \BB^l}
{\tilde S}(z)^{{\bf a}, {\bf b}}_{{\bf i}, {\bf j}}\,v_{\bf a} \otimes v_{{\bf b}}
\qquad ({\bf i} \in \BB^k, {\bf j} \in \BB^l),
\label{S01}
\\
{\tilde S}(z)^{{\bf a}, {\bf b}}_{{\bf i}, {\bf j}} & = (-q)^{k-l+\eta}
R(q^{l-k}z)^{{\bf a}, {\bf b}}_{{\bf i}, {\bf j}} \left| _{q\rightarrow t^{1/2}}\right.,
\label{S02}
\\
\eta &= \eta^{{\bf a}, {\bf b}}_{{\bf i}, {\bf j}} = \sum_{0 \le r<s \le n}
(b_ra_s-i_rj_s).
\label{S03}
\end{align}
The quantity $\eta$  \eqref{S03} is formally the same as \cite[eq. (16)]{KMMO16}.
Obviously, ${\tilde S}(z)$ also possesses the weight conservation property as \eqref{rwc}.
Moreover, the Yang-Baxter equation \eqref{ybe1} for $R^{k,l}(z)$
and the same argument as in the proof of \cite[Prop.4]{KMMO16}
imply that ${\tilde S}^{k,l}(z)$ also satisfies the Yang-Baxter equation:
\begin{align}\label{ybe2}
{\tilde S}^{k_1,k_2}_{1,2}(x){\tilde S}^{k_1,k_3}_{1,3}(xy){\tilde S}^{k_2,k_3}_{2,3}(y) = 
{\tilde S}^{k_2,k_3}_{2,3}(y) {\tilde S}^{k_1,k_3}_{1,3}(xy){\tilde S}^{k_1,k_2}_{1,2}(x),
\end{align}
where $k_1,k_2,k_3 \in \{0,\ldots, n+1\}$.

The $t$-PushTASEP will be related to the $l=1$ case of  ${\tilde S}^{k,l}(z)$. 
For convenience we introduce a slight overall renormalization of them as 
$S(z) = S^{k,1}(z) = (1-z)(1-tz){\tilde S}^{k,1}(z)$.
Explicitly, we set 
\begin{align}
S(z)(v_{\bf i} \otimes v_{{\bf e}_j}) &=\sum_{{\bf a}\in \BB^k,  {\bf e}_b \in \BB^1}
S(z)^{{\bf a}, {\bf e}_b}_{{\bf i}, {\bf e}_j}\,v_{\bf a} \otimes v_{{\bf e}_b}
\qquad ({\bf i} \in \BB^k, {\bf e}_j \in \BB^1),
\label{S1}
\\
S(z)^{{\bf a}, {\bf e}_b}_{{\bf i}, {\bf e}_j} & = 
(1-z)(1-q^2z)(-q)^{k-1+\eta}
R(q^{1-k}z)^{{\bf a}, {\bf e}_b}_{{\bf i}, {\bf e}_j} \left| _{q\rightarrow t^{1/2}}\right.,
\label{S2}
\\
\eta &= \eta^{{\bf a}, {\bf e}_b}_{{\bf i}, {\bf e}_j} = \sum_{s>b}a_s-\sum_{r<j}i_r,
\label{S3}
\end{align}
where, from \eqref{rwc}, we assume ${\bf a}+{\bf e}_b = {\bf i} + {\bf e}_j$.
This leads to the expression
\begin{align}
k-1+\eta = \begin{cases}
a_j-1+2(a_{j+1}+\cdots + a_n), & j=b,
\\
-(a_{j+1}+\cdots + a_{b-1})+2(a_{j+1}+\cdots + a_n), & j<b,
\\
(a_{b+1}+\cdots + a_{j-1})+2(a_{j+1}+\cdots + a_n), & j>b.
\end{cases}
\end{align}
Therefore the matrix element \eqref{S2} takes the form, which reproduces \eqref{Srepeat0}:
\begin{align}
S(z)^{{\bf a}, {\bf e}_b}_{{\bf i}, {\bf e}_j} &= \delta^{{\bf a}+{\bf e}_b}_{{\bf i} + {\bf e}_j}
(-1)^{a_0+\cdots + a_{j-1}+i_0+\cdots + i_{b-1}}
 t^{a_{j+1}+\cdots + a_n}(1-t^{a_j}z^{\delta_{b,j}})z^{[j>b]}.
 \label{S}
\end{align}
It is noteworthy that \eqref{S} is a polynomial in both $t$ and $z$.

When $k=1$, the nonzero elements of $S^{1,1}(z)$ are limited to the form 
$S(z)^{{\bf e}_a, {\bf e}_b}_{{\bf e}_i, {\bf e}_j}$ where  $a,b,i,j \in \{0,\ldots, n\}$ 
and $\{a,b\}=\{i,j\}$ as multisets.
Explicitly they are given by 
\begin{align}\label{Sb}
S(z)^{{\bf e}_i, {\bf e}_i}_{{\bf e}_i, {\bf e}_i} = 1-tz,
\quad
S(z)^{{\bf e}_i, {\bf e}_j}_{{\bf e}_i, {\bf e}_j} = (1-z)t^{[i>j]} \; (i\neq j),
\quad
S(z)^{{\bf e}_j, {\bf e}_i}_{{\bf e}_i, {\bf e}_j} = (1-t)z^{[i<j]}\; (i\neq j).
\end{align}
They coincide with \eqref{Sb2}.

\section{$H_{\text{PushTASEP}}(x_1,\ldots, x_L)$ from transfer matrices for symmetric fusion}
\label{app:sym}
 
 For comparison, we briefly sketch an alternative description of $\mathcal{H}$ 
 \eqref{cH} in terms of 
 transfer matrices corresponding to symmetric fusion. 
 We introduce the symmetric tensor counterparts of \eqref{Bk} and \eqref{Vk}. 
 For $k \in \Z_{\ge 0}$, define
\begin{align}
\BB_k &= \{{\bf i} = (i_0,\ldots, i_n) \in (\Z_{\ge 0})^{n+1}\mid |{\bf i}|=k\},
\quad  (|{\bf i}|=i_0+\cdots + i_n),
\label{Bsk}\\
V_k &= \bigoplus_{{\bf i} \in \BB_k} \mathbb{C}v'_{\bf i}.
\label{Vsk}
\end{align}
In the special case $k=1$, we identify $V_1$ with $\VV$, the space of local states of the 
$t$-PushTASEP defined in Section\ref{sec:tpush}, following the same rule as \eqref{v11},  with 
$v'_{{\bf e}_i} \in V_1$ replacing $v'_{{\bf e}_i} \in V^1$. That is,
 \begin{align}\label{vs11}
\VV \ni \vv_i = v'_{{\bf e}_i} \in V_1 \;\;\text{where}\; \;{\bf e}_i=(\delta_{0,i},\ldots, \delta_{n,i}) \in \BB_1
\quad (0 \le i \le n).
\end{align}

A quantum $R$-matrix acts on $V_k \otimes V_1$ and satisfies the Yang-Baxter equation.
To distinguish it from $S(z)=S^{k,1}(z)$ in \eqref{S1}--\eqref{S}, we use a a different notation:
$\mathscr{S}(z)= \mathscr{S}_{k,1}(z)$.
Employing a suitable gauge (cf. \cite[eq. (86)]{KOS24}), we have 
\begin{align}
\mathscr{S}(z)(v'_{\bf i} \otimes v'_{{\bf e}_j}) &=\sum_{{\bf a}\in \BB_k,  {\bf e}_b \in \BB_1}
\mathscr{S}(z)^{{\bf a}, {\bf e}_b}_{{\bf i}, {\bf e}_j}\,v'_{\bf a} \otimes v'_{{\bf e}_b}
\qquad ({\bf i} \in \BB_k, {\bf e}_j \in \BB_1),
\label{St1}
\\
\mathscr{S}(z)^{{\bf a}, {\bf e}_b}_{{\bf i}, {\bf e}_j} &= \delta^{{\bf a}+{\bf e}_b}_{{\bf i} + {\bf e}_j}
 t^{i_{b+1}+\cdots + i_n}(1-t^{i_b}z^{\delta_{b,j}})z^{[j>b]}.
\label{St}
\end{align}
The most significant difference from \eqref{S} is that 
\eqref{St}  is defined for ${\bf a}, {\bf i} \in \BB_k$ \eqref{Bsk} rather than $\BB^k$ \eqref{Bk}.

For $k \in \Z_{\ge 0}$, let $T_k(z) = T_k(z|x_1,\ldots, x_L)$ be the transfer matrix whose auxiliary space is $V_k$.
Following the construction in \eqref{tkdef}-- \eqref{tke}, it is given by 
\begin{subequations}
\begin{align}
T_k(z) |\sigma_1,\ldots, \sigma_L\rangle  &= \sum_{\sigma'_1,\ldots, \sigma'_L \in \{0,\ldots, n\}}
T_k(z)^{\sigma'_1,\ldots, \sigma'_L}_{\sigma_1,\ldots, \sigma_L}
|\sigma'_1, \ldots, \sigma'_L\rangle,
\label{tskdef}\\
T_k(z)^{\sigma'_1,\ldots, \sigma'_L}_{\sigma_1,\ldots, \sigma_L}
&= \sum_{{\bf a}_1,\ldots, {\bf a}_L \in \BB_k}
\mathscr{S} \Bigl( \frac{z}{x_1} \Bigr)^{{\bf a}_2, {\bf e}_{\sigma'_1}}_{{\bf a}_1, {\bf e}_{\sigma_1}}
\mathscr{S} \Bigl( \frac{z}{x_2} \Bigr)^{{\bf a}_3, {\bf e}_{\sigma'_2}}_{{\bf a}_2, {\bf e}_{\sigma_2}} 
\cdots 
\mathscr{S} \Bigl( \frac{z}{x_L} \Bigr)^{{\bf a}_1, {\bf e}_{\sigma'_L}}_{{\bf a}_L, {\bf e}_{\sigma_L}}.
\label{tske}
\end{align}
\end{subequations}
We note the relations $\mathscr{S}_{1,1}(z) = S^{1,1}(z)$, as well as 
$T_0(z) = T^0(z)$ and $T_1(z)=T^1(z)$.
By the Yang-Baxter equation, the commutativity holds:
\begin{subequations}
\begin{align}
[T_k(z|x_1,\ldots, x_L), T_{k'}(z'|x_1,\ldots, x_L)]&=0
\qquad  (k,k' \in \Z_{\ge 0}),
\\
[T_k(z|x_1,\ldots, x_L), T^{l}(z'|x_1,\ldots, x_L)]&=0
\qquad (k\in \Z_{\ge 0}, l \in \{0,\ldots, n+1\})
\end{align}
\end{subequations}
in addition to \eqref{tcom}.

The transfer matrices $T_k(z)$ and $T^k(z)$ serve as spectral parameter dependent analogues of the 
completely symmetric polynomial $h_k$ and the elementary symmetric polynomial $e_k$, respectively.
They satisfy various functional relations.
For instance, the following Jacobi-Trudi type formula holds (cf. \cite[Th.6.1, 6.2]{KNS11}):
\begin{subequations}
\begin{align}
T_l(z) &=\left(T^0(tz)T^0(t^2z)\cdots T^0(t^{l-1}z)\right)^{-1}
\mathrm{det}\left(T^{1-i+j}(t^{j-1}z)\right)_{1 \le i,j \le l},
\label{jt1}
\\
T^k(z) &= \left(T^0(t^{-1}z)T^0(t^{-2}z)\cdots T^0(t^{-k+1}z)\right)^{-1}
\mathrm{det}\left(T_{1-i+j}(t^{1-j}z)\right)_{1 \le i,j \le k},
\label{jt2}
\end{align}
\end{subequations}
where $T^0(z)$ is a scalar matrix \eqref{T0z}, and we set $T^k(z)=0$ for $k<0$ or $k>n+1$. 
Substituting \eqref{jt2} into $\mathcal{H}$ in \eqref{cH} provides an alternative expression for 
$H_\text{PushTASEP}$ in terms of 
differential coefficients of $T_0(z), T_1(z), \ldots, T_n(z)$.
However the resulting formula is not particularly illuminating.
For instance for $n=2$, we obtain 
\begin{subequations}
\begin{align}
\mathcal{H} =&\frac{\dot{T}_2(0)
-(1 + t^{m_0} + t^{1 + m_0} + t^{m_0 + m_1} + t^{1 + m_0 + m_1})\dot{T}_1(0)
+tC \sum_{j=1}^L\frac{1}{x_j}}{(1 - t) t (1 - t^{m_0}) (1 - t^{m_0 + m_1})},
\\
C =& -1 + t - t^{-1 + m_0} - t^{2 m_0} - t^{1 + m_0} - t^{-1 + m_0 + m_1} 
        - t^{2 (m_0 + m_1)} 
        \\
        &- t^{1 + m_0 + m_1} - t^{-1 + 2 m_0 + m_1} 
        - 2 t^{2 m_0 + m_1},
\end{align}
\end{subequations}
where $\dot{T}_l(0)= \left. \frac{dT_l(z)}{dz}\right|_{z=0}$.

\section*{Acknowledgments}

The authors would like to thank the organizers of the workshop, 
{\em Discrete integrable systems: difference equations, cluster algebras and probabilistic models}
at the  International Centre for Theoretical Sciences, Bengaluru, India 
from October 21 to November 1,  2024, for their kind invitation and warm hospitality, where  
this work was initiated.
A. A. was partially supported by SERB Core grant CRG/2021/001592 as well as the DST FIST program - 2021 [TPN - 700661]. 
A.K. is supported by Grants-in-Aid for Scientific Research No.~24K06882 from JSPS.

\end{document}